\documentclass[runningheads]{llncs}

\usepackage[T1]{fontenc}

\usepackage{graphicx}

\usepackage{xcolor, colortbl}

\usepackage{adjustbox}

\usepackage[framemethod=TikZ]{mdframed}
\usepackage{etoolbox}

\usepackage[frozencache]{minted}

\usepackage{stmaryrd}

\usepackage{mathtools}
\usepackage{amsmath} 
\usepackage{amssymb} 
\usepackage{mathpartir}

\usepackage{pifont}

\usepackage{algorithm2e}

\usepackage{diagbox}

\usepackage{hyperref}

\usepackage{stmaryrd}

\usepackage[
labelfont=sf,
hypcap=false,
format=hang,
width=0.8\columnwidth
]{caption}

\usepackage{tikz}
\usetikzlibrary{graphs}
\usetikzlibrary{positioning}

\AtBeginDocument{%
  }

\definecolor{egreen}{rgb}{0, 0.4, 0.267}
\definecolor{dkviolet}{rgb}{0.6,0,0.8}
\definecolor{dkgreen}{rgb}{0,0.4,0}
\definecolor{dkblue}{rgb}{0,0.1,0.5}
\definecolor{lightblue}{rgb}{0,0.5,0.5}
\definecolor{orange}{rgb}{0.9,0.39,0}
\definecolor{lightgrey}{RGB}{240,240,240}
\definecolor{darkgrey}{RGB}{125,125,125}

\newcommand{\orange}[1]{\textcolor{orange}{#1}}
\newcommand{\egreen}[1]{\textcolor{egreen}{#1}}

\newcommand{\enzo}[1]{\orange{}}
\newcommand{\assia}[1]{\egreen{}}

\newcommand{\ie}{\textit{i.e.}, }

\newcommand{\eg}{\textit{e.g.}, }
\newcommand{\etal}{\textit{et al.}}

\newcommand{\mono}[1]{\texttt{#1}}

\newcommand{\ninferrule}[3]{\inferrule{#1}{#2}\,(\textsc{#3})}
\newcommand{\subtype}{\ensuremath{\preccurlyeq}}
\newcommand{\term}[1]{\ensuremath{\mathcal{T}_{#1}}}
\newcommand{\constants}{\ensuremath{\mathcal{C}}}
\newcommand{\constantTypes}[1]{\ensuremath{T_{#1}}}
\newcommand{\eraseAnn}[1]{\ensuremath{\left|\,#1\,\right|^-}}
\newcommand{\topAnn}[1]{\ensuremath{\left|\,#1\,\right|^+}}

\newcommand{\ann}{\ensuremath{\mathcal{A}}}

\newcommand{\Var}{\mathrm{Var}}

\newcommand{\CComega}{\ensuremath{CC_\omega}}
\newcommand{\CCann}{\ensuremath{CC_\omega^+} }

\newcommand{\Trocq}{\textsc{Trocq}}
\newcommand{\Coq}{\textsf{Coq}}
\newcommand{\Agda}{\textsf{Agda}}
\newcommand{\CubicalAgda}{\textsf{Cubical Agda}}
\newcommand{\Lean}{\textsf{Lean}}
\newcommand{\mathlib}{\textsf{mathlib}}
\newcommand{\IsabelleHOL}{\textsf{Isabelle/HOL}}
\newcommand{\ArchiveOfFormalProofs}{\textsf{Archive of Formal Proofs}}
\newcommand{\coqinterval}{\textsf{coq-interval}}
\newcommand{\mathcompanalysis}{\textsf{mathcomp-analysis}}
\newcommand{\HoTT}{\textsf{HoTT}}

\newcommand{\CoqElpi}{\textsf{Coq-Elpi}}
\newcommand{\CoqEAL}{\textsf{CoqEAL}}

\newcommand{\MetaCoq}{\textsf{MetaCoq}}
\newcommand{\LtacII}{\textsf{Ltac 2}}
\newcommand{\Mtac}{\textsf{Mtac}}
\newcommand{\Elpi}{\textsf{Elpi}}
\newcommand{\LambdaProlog}{$\lambda$\textsf{Prolog}}
\newcommand{\Prolog}{\textsf{Prolog}}
\newcommand{\PumpkinPi}{\textsf{Pumpkin Pi}}

\newcommand{\coq}[1]{\mintinline{coq}`#1`}
\newcommand{\coqescape}[1]{\mintinline[mathescape=true,escapeinside=//]{coq}`#1`}
\newcommand{\elpi}[1]{\mintinline{elpi}`#1`}

\newcommand{\param}[3]{#1\,\sim\,#2\ \because\ #3}%
\newcommand{\paramq}[4]{#1\ @\ #2\,\sim\,#3\ \because\ #4}%

\newcommand{\paramt}[1]{\llbracket\,#1\,\rrbracket}
\newcommand{\paramtm}[1]{[\,#1\,]}
\newcommand{\relofparamtm}[1]{\textsf{rel}([\,#1\,])}
\newcommand{\txta}[0]{\text{a}}
\newcommand{\txtb}[0]{\text{b}}

\newcommand{\adm}[2]{{#1}\rhd{#2}}%

\newcommand{\ParamRec}[1]{\ensuremath{\boxbox}^{#1}}
\newcommand{\pproof}{\ensuremath{p}}

\newcommand{\ParamClass}[1]{\textsf{Class}_{#1}}
\newcommand{\ParamM}[1]{\textsf{M}_{#1}}
\newcommand{\Rel}[1]{\textsf{rel}(#1)}
\newcommand{\map}[1]{\textsf{map}(#1)}
\newcommand{\comap}[1]{\textsf{comap}(#1)}

\newcommand{\Type}[1]{\ensuremath{\square_{#1}}}
\newcommand{\emptyctxt}{\langle\rangle}
\newcommand{\tyju}[3]{\ #1 \vdash #2 : #3}

\newcommand{\rsym}[1]{{#1}^{-1}}
\newcommand{\dep}{\mathcal{D}}
\newcommand{\tonat}{\ensuremath{\uparrow_\mathbb{N}}}
\newcommand{\ofnat}{\ensuremath{\downarrow_\mathbb{N}}}

\newcommand{\ptweq}{\doteqdot}
\newcommand{\tequiv}{\simeq}
\newcommand{\isequiv}[1]{\textsf{IsEquiv}(#1)}
\newcommand{\isiso}[1]{\textsf{IsIso}(#1)}
\newcommand{\isfun}[1]{\textsf{IsFun}(#1)}
\newcommand{\iscontr}[1]{\textsf{IsContr}(#1)}
\newcommand{\isumap}[1]{\textsf{IsUmap}(#1)}
\newcommand{\conv}{\equiv}

\newcounter{tmpFigure}
\newcounter{nbListings}

\newcounter{nbAlgorithms}

\newcommand{\xnnR}{\ensuremath{{\overline{\mathbb{R}}_{\geq 0}}}}
\newcommand{\nnR}{\ensuremath{{\mathbb{R}_{\geq 0}}}}

\mdfdefinestyle{listingstyle}{%
  backgroundcolor=black!5,hidealllines=true,skipabove=5pt%
  innerleftmargin=5pt,innerrightmargin=0pt,innertopmargin=2pt,innerbottommargin=0pt%
}

\newcommand\trocqcode[2]{\href{https://github.com/coq-community/trocq/blob/0.1.5/#1}{#2}}

\makeatletter
\newcommand*\mysize{%
  \@setfontsize\mysize{9.0}{9.0}%
}
\makeatother
\setminted[coq]{bgcolor=lightgrey,mathescape=true,escapeinside=//,fontsize=\mysize}

\begin{document}


\title{Trocq: Proof Transfer for Free, With or Without Univalence\thanks{This work has received funding from the {European Research Council (ERC)} under the European Union’s Horizon 2020 research and innovation programme (grant agreement No.{101001995}).}}


\author{Cyril Cohen\inst{1}\orcidID{0000-0003-3540-1050} \and
Enzo Crance\inst{2,3}\orcidID{0000-0002-0498-0910} \and
Assia Mahboubi\inst{2,4}\orcidID{0002-0312-5461}}

\institute{Université Côte d'Azur, Inria\\
\email{Cyril.Cohen@inria.fr} \and
Nantes Université, École Centrale Nantes, CNRS, INRIA, LS2N, UMR 6004
\email{\{Enzo.Crance,Assia.Mahboubi\}@inria.fr} \and
Mitsubishi Electric R\&D Centre Europe
\and
Vrije Universiteit Amsterdam, The Netherlands
}

\maketitle
\begin{abstract}
This article presents \Trocq{}, a new proof transfer framework for dependent type theory. \Trocq{} is
based on a novel formulation of type equivalence, used to generalize the univalent
parametricity translation. This framework takes care of avoiding dependency on the axiom of
univalence when possible, and may be used with more relations than just equivalences.
We have implemented a corresponding plugin for the \Coq{} interactive theorem prover, in the
\CoqElpi{} meta-language.

\end{abstract}



\section{Introduction}\label{sec:intro}

Formalizing mathematics provides every object and statement of the mathematical literature with an explicit data structure, in a
certain choice of foundational formalism. As one would expect, several
such explicit representations are most often needed for a same
mathematical concept. Sometimes, these different choices are already
made explicit on paper: multivariate polynomials can for instance be
represented as lists of coefficient-monomial pairs, \eg when
computing Gröbner bases, but also as univariate polynomials with
polynomial coefficients, \eg for the purpose of projecting algebraic
varieties. The conversion between these equivalent data structures
however remains implicit on paper, as they code in fact for the
same free commutative algebra. In some other cases, implementation
details are just ignored on paper, \eg when a proof involves both
reasoning with Peano arithmetic and computing with large
integers.
\begin{example}[Proof-oriented vs.
    computation-oriented data structures]\label{ex:elim}
  The standard library of the \Coq{} interactive theorem
  prover~\cite{the_coq_development_team_2022_7313584} has
  two data structures for representing natural numbers. Type $\mathbb{N}$ is the base-1 number system and the associated elimination principle \coqescape{/$\mathbb{N}$/_ind} is the usual recurrence scheme:
\begin{minted}{coq}
Inductive /$\mathbb{N}$/ : Type := O/$_{\mathbb{N}}$/ : /$\mathbb{N}$/ | /S$_{\mathbb{N}}$/ (n : /$\mathbb{N}$/) : /$\mathbb{N}$/.

/$\mathbb{N}$/_ind : /$\forall$/ P : /$\mathbb{N}$/ /$\rightarrow$/ /$\square$/, P O/$_{\mathbb{N}}$/ /$\rightarrow$/ (/$\forall$/ n : /$\mathbb{N}$/, P n /$\rightarrow$/ P (S n)) /$\rightarrow$/ /$\forall$/ n : /$\mathbb{N}$/, P n
\end{minted}
On the other hand, type \coq{N} provides a binary representation \coq{positive} of non-negative integers, as sequences of bits with a head $1$,
and is thus better suited for coding efficient arithmetic
operations. The successor function
\coq{S/$_{\tt N}$/ : N /$\to$/ N} is no longer a constructor of the type, but can
be implemented as a program, via an auxiliary successor function \coq{S/$_{{\tt pos}}$/}
for type \coq{positive}.
\begin{minted}{coq}
Inductive positive : Type :=
  xI : positive /$\to$/ positive | xO : positive /$\to$/ positive | xH : positive.

Inductive N : Type := O/$_{\texttt{N}}$/ : N | Npos : positive /$\rightarrow$/ N.

Fixpoint S/$_\mono{pos}$/ (p : positive) : positive := match p with
  | xH /$\Rightarrow$/ xO xH | xO p /$\Rightarrow$/ xI p | xI p /$\Rightarrow$/ xO (S/$_\mono{pos}$/ p) end.

Definition S/$_{\texttt{N}}$/ (n : N) := match n with
  | Npos p /$\Rightarrow$/ Npos (S/$_{\texttt{pos}}$/ p) | _ /$\Rightarrow$/ Npos xH end.
\end{minted}
This successor function is useful to implement conversions
$\tonat\ :\ \texttt{N}\ \to\ \mathbb{N}$ and
$\ofnat\ :\ \mathbb{N}\ \to\ \texttt{N}$
between the unary and binary
representations. These conversion functions are in fact inverses of
each other. The natural recurrence scheme on natural numbers thus
\emph{transfers} to type \coq{N}:
\begin{minted}{coq}
N_ind : /$\forall$/ P : N /$\rightarrow$/ /$\square$/, P O/$_{\texttt{N}}$/ /$\rightarrow$/ (/$\forall$/ n : N, P n /$\rightarrow$/ P (S/$_{\texttt{N}}$/ n)) /$\rightarrow$/ /$\forall$/ n : N, P n
\end{minted}
Incidentally, \coq{N_ind} can be proved from \coq{/$\mathbb{N}$/_ind}
by using only the fact that $\ofnat$ is a left inverse
of $\tonat$, and the following compatibility lemmas:
\[ \ofnat \texttt{O}_{\mathbb{N}} = \texttt{O}_{\texttt{N}} \quad \text{and} \quad \forall n : \mathbb{N},\quad \ofnat (\texttt{S}_{\mathbb{N}}\ n) = \texttt{S}_\texttt{N}\ (\ofnat\,n)\]
\end{example}

Proof transfer issues are not tied to program verification. For
instance, the formal study of summation and integration, in basic real
analysis, provides a classic example of frustrating 
bureaucracy.
\begin{example}[Extended domains]\label{ex:sums}
Given a sequence $(u_n)_{n\in\mathbb{N}}$ of non-negative
real numbers, \ie a function
$u\ : \ \mathbb{N} \rightarrow [0,+\infty[$, $u$ is said to be
\emph{summable} when the sequence $(\sum_{k=0}^nu_k)_{n\in\mathbb{N}}$ has a finite limit,
denoted $\sum u$. Now for two summable sequences $u$ and $v$, it
is easy to see that $u + v$, the sequence obtained by point-wise addition of $u$ and
$v$, is also a summable sequence, and that:
\begin{equation}\label{eq:sum}
\sum (u + v) = \sum u + \sum v
\end{equation}
As expression $\sum u$ only makes sense when $u$ is a summable sequence, any algebraic operation ``under the sum'', \eg rewriting $\sum (u + (v + w))$ into $\sum ((w + u) + v)$, \emph{a priori} requires a proof of summability for every rewriting step. In a classical setting, the standard approach rather assigns a default value to the case of an infinite sum, and introduces an extended domain $[0,+\infty]$. Algebraic operations on real numbers, like addition, are extended to the extra $+\infty$ case. Now for a sequence $u\ : \ \mathbb{N} \rightarrow [0,+\infty]$, the limit $\sum u$ is always defined, as
increasing partial sums either converge to a finite limit, or diverge
to $+\infty$. The road map is then to first prove that
Equation~\ref{eq:sum} holds for \emph{any} two sequences of
\emph{extended} non-negative numbers. The result is then
\emph{transferred} to the special case of summable sequences of
non-negative numbers. Major libraries of formalized mathematics
including \Lean{}'s \mathlib{}~\cite{DBLP:conf/cpp/X20}, \IsabelleHOL{}'s
\ArchiveOfFormalProofs{},
\coqinterval{}~\cite{DBLP:journals/jar/Martin-DorelM16} or \Coq{}'s
\mathcompanalysis{}~\cite{affeldt2023measure}, resort to such
extended domains and transfer steps, notably for defining measure
theory.  Yet, as reported by expert users~\cite{gouezel}, the
associated transfer bureaucracy is essentially done manually and thus
significantly clutters formal developments in real and complex
analysis, probabilities,~etc.
\end{example}

Users of interactive
theorem provers should be allowed to elude mundane arguments
pertaining to proof transfer, as they would on paper, and spare
themselves the related bureaucracy. Yet, they
still need to convince the proof checker and thus have to provide
explicit transfer proofs, albeit ideally automatically generated
ones. The present work aims at providing a general method for
implementing this nature of automation, for a diverse range of proof
transfer problems.

In this paper, we focus on interactive theorem provers based on dependent type theory,
such as \Coq{}, \Agda{}~\cite{DBLP:conf/afp/Norell08} or
\Lean{}~\cite{DBLP:conf/cade/Moura021}. These proof management systems
are genuine functional programming languages, with full-spectrum
dependent types, a context in which representation independence
meta-theorems can be turned into concrete instruments for achieving
program and proof transfer.

Seminal results on the contextual equivalence of distinct implementations of a
same abstract interface were obtained for System F, using logical
relations~\cite{DBLP:conf/popl/Mitchell86} and parametricity
meta-theorems~\cite{DBLP:conf/ifip/Reynolds83,DBLP:conf/fpca/Wadler89}.
In the context of type theory, such meta-theorems can be turned into
syntactic translations of the type theory of interest into itself,
automating this way the generation of the statement and proof of
parametricity properties for type families and for programs. Such syntactic
relational models can accommodate dependent
types~\cite{DBLP:conf/fossacs/BernardyL11}, inductive
types~\cite{DBLP:journals/jfp/BernardyJP12} and scale to the 
Calculus of Inductive Constructions, with an impredicative
sort~\cite{DBLP:conf/csl/KellerL12}.

In particular, the
\emph{univalent parametricity} translation~\cite{univparam2} leverages the univalence axiom~\cite{hottbook} so as to transfer
statements using established equivalences of types. This
approach crucially removes the need for devising an explicit common
interface for the types in relation. In presence of an internalized
univalence axiom and of higher-inductive types, the \emph{structure
identity principle} provides internal representations of independence
results, for more general relations between types
than equivalences~\cite{DBLP:journals/pacmpl/AngiuliCMZ21}. This last
approach is thus particularly relevant in cubical type
theory~\cite{cohen_et_al:LIPIcs.TYPES.2015.5,DBLP:journals/pacmpl/VezzosiM019}. Indeed, a computational interpretation of the
univalence axiom brings computational adequacy to otherwise possibly
stuck terms, those resulting from a transfer involving an axiomatized
univalence principle.

Yet taming the bureaucracy of proof transfer remains hard in practice for users of \Coq{}, \Lean{} or \Agda{}. Examples~\ref{ex:elim} and~\ref{ex:sums} actually illustrate fundamental limitations of the existing approaches:

\paragraph{Univalence is overkill} Both univalent parametricity and the structure identity principle can be used to derive the statement and the proof of the induction principle \coq{N_ind} of Example~\ref{ex:elim}, from the elimination scheme of type $\mathbb{N}$. But up to our knowledge, all the existing methods for automating this implication pull in the univalence principle in the proof, although it can be obtained by hand by very elementary means. This limitation is especially unsatisfactory for developers of libraries formalizing classical mathematics, and notably \Lean{}'s
\mathlib{}. These libraries indeed typically assume a strong form of proof irrelevance, which is incompatible with univalence, and thus with univalent parametricity.

\paragraph{Equivalences are not enough, neither are quotients}
Univalent parametricity cannot help with Example~\ref{ex:sums}, as type $[0,+\infty[$ is \emph{not equivalent} to its extended version $[0,+\infty]$. In fact, we are not aware of any tool able to automate this proof transfer. In particular, the structure identity principle~\cite{DBLP:journals/pacmpl/AngiuliCMZ21} would not apply as such.

\paragraph{Contributions} In short, existing techniques for
transferring results from one type to another, \eg from \coq{/$\mathbb{N}$/} to \coq{N} or from extended real numbers to real numbers, are either not
suitable for dependent types, or too coarse to track the
exact amount of data needed in a given proof, and not more.
This paper presents three contributions improving this unfortunate state of affairs:
\begin{itemize}
\item A parametricity framework \emph{à la carte}, that generalizes
  the univalent parametricity translation~\cite{univparam2}, as well as refinements à la
  \CoqEAL{}~\cite{DBLP:conf/cpp/CohenDM13} and generalized
  rewriting~\cite{DBLP:journals/jfrea/Sozeau09}.
  Its pivotal ingredient is a variant of Altenkirch and Kaposi's symmetrical presentation of type equivalence~\cite{DBLP:conf/types/AltenkirchK15}.
\item A conservative subtyping extension of $\CComega$~\cite{DBLP:journals/iandc/CoquandH88}, used to
  formulate an inference algorithm for the synthesis of parametricity
  proofs.
\item The implementation of a new parametricity plugin for the \Coq{}
  interactive theorem prover, using the \CoqElpi{}~\cite{tassi:hal-01897468}
  meta-language. This plugin rests on original formal
  proofs, conducted on top of the \HoTT{}
  library~\cite{DBLP:conf/cpp/BauerGLSSS17}, and is distributed with
  a collection of application examples.
\end{itemize}

\paragraph{Outline}
The rest of this paper is organized as follows.
Section~\ref{sec:strenghts-limits} introduces proof transfer and recalls
the principle, strengths and weaknesses of the univalent parametricity
translation.
In Section~\ref{sec:disassembling-reassembling}, we present a new
definition of type equivalence, motivating a hierarchy of structures for relations preserved by parametricity.
Section~\ref{sec:strat-param-rules} then presents variants of parametricity translations.
In Section~\ref{sec:applications}, we discuss a few examples of applications and we conclude in
Section~\ref{sec:rel-work-and-concl}.

\section{Strengths and limits of univalent parametricity}\label{sec:strenghts-limits}

We first clarify the essence of proof transfer in dependent type theory (\textsection~\ref{ssec:bigp}) and briefly recall a few concepts related to type equivalence and univalence (\textsection~\ref{ssec:univ}). We then review and discuss the limits
of univalent parametricity (\textsection~\ref{ssec:univparam}).

\subsection{Proof transfer in type theory}\label{ssec:bigp}
We recall the syntax of the Calculus of Constructions, \CComega{}, a
$\lambda$-calculus with dependent function types and a predicative
hierarchy of universes, denoted $\Type{i}$:

\[A, B, M, N ::= \Type{i}\ |\ x\ |\ M\ N \ |\ \lambda x : A.\,M
\ |\ \Pi x : A.\,B\]
We omit the typing rules of the calculus, and refer the reader to standard references (\eg~\cite{paulinmohring:hal-01094195,nederpelt_geuvers_2014}). We also use the standard
equality type, called propositional equality, as well as dependent pairs,
denoted $\Sigma x : A.\,B$.  We write $t \conv u$ the definitional equality
between two terms $t$ and $u$. Interactive theorem provers like \Coq{}, \Agda{} and \Lean{} are
based on various extensions of this core, notably with inductive types or with
an impredicative sort. When the universe level does not matter, we casually
remove the annotation and use notation $\Type{}$.

In this context, proof transfer from type $T_1$ to type $T_2$ roughly
amounts to \emph{synthesizing} a new type former $W : T_2 \rightarrow
\Type{}$, \ie a type parametric in some type $T_2$, from an initial
type former $V : T_1 \rightarrow \Type{}$, \ie a type parametric
in some type $T_1$, so as to ensure that for some given relations
$R_{T} : T_1 \rightarrow T_2 \rightarrow \Type{}$ and
$R_{\Type{}} : \Type{} \rightarrow \Type{} \rightarrow \Type{}$,
there is a proof $w$ that:
\[\tyju{\Gamma}{w}{\forall (t_1 : T_1) (t_2 : T_2),
  R_{T}\ t_1\ t_2 \rightarrow R_{\Type{}}(V\ t_1) (W\ t_2)}\]
for a
suitable context $\Gamma$. This setting generalizes as expected to
$k$-ary type formers, and to more pairs of related types. In practice,
relation $R_{\Type{}}$ is often a right-to-left arrow, \ie
$R_{\Type{}}\ A\ B \triangleq B \rightarrow A$, as in this case the
proof $w$ substantiates a proof step turning a goal clause $\Gamma
\vdash V\ t_1$ into $\Gamma \vdash W\ t_2$.

Phrased as such, this
synthesis problem is arguably quite loosely specified. Consider for instance
the transfer problem discussed in Example~\ref{ex:elim}. A first possible
formalization involves the design of an appropriate common interface structure
for types  $\mathbb{N}$ and $\texttt{N}$, for instance by setting both $T_1$ and $T_2$ as
$\Sigma N : \Type{}. N \times (N \to N)$, and both $V$ and $W$ as:
$\lambda X : T_1.\, \Pi P : X.1 \to \Box.\, P\ X.2 \to (\Pi n : X.1.\, P\ n \to P\ (X.3\ n)) \to \Pi n : X.1.\, P\ n$,
where $X.i$ denotes the $i$-th item in the dependent tuple $X$. In this case,
relation $R_T$ may characterize isomorphic instances of the structure. Such instances of
proof transfer are elegantly addressed in cubical type theories via internal representation independence results~\cite{DBLP:journals/pacmpl/AngiuliCMZ21}. In the context of \CComega, the hassle of devising explicit
structures by hand has been termed the \emph{anticipation} problem~\cite{univparam2}.

Another option is to consider two different types $T_1 \triangleq \mathbb{N}
\times (\mathbb{N} \rightarrow \mathbb{N})$ and $T_2 \triangleq \mathrm{N}
\times (\mathrm{N} \rightarrow \mathrm{N})$ and
\begin{align*}
  V' \triangleq\ & \lambda X : T_1.\  \forall P : \mathbb{N} \to \Box.\ P\ X.1 \to (\forall n : \mathbb{N}, P\ n \to P (X.2\ n)) \to \forall n : \mathbb{N}, P\ n\\
  W' \triangleq\ & \lambda X : T_2.\  \forall P : \mathrm{N} \to \Box.\ P\ X.1 \to (\forall n : \mathrm{N}, P\ n \to P (X.2\ n)) \to \forall n : \mathrm{N}, P\ n
\end{align*}
where one
would typically expect $R_T$ to be a type equivalence between $T_1$ and $T_2$, so as to
transport $(V'\ t_1)$ to $(W'\ t_2)$, along this equivalence.

Note that some solutions of given instances of proof transfer problems are
in fact too trivial to be of interest. Consider for example the case
of a \emph{functional} relation between $T_2$ and $T_1$,
with $R_T \ t_1\ t_2$ defined as $t_1 = \phi\ t_2$,
for some $\phi : T_2 \rightarrow T_1$.
In this case, the composition $V \circ \phi$ is an obvious candidate for $W$,
but is often uninformative. Indeed, this composition can only
propagate structural arguments, blind to the additional
mathematical proofs of program equivalences potentially available
in the context. For instance, here is a useless variant of $W'$:
\begin{align*}
  W'' \triangleq \lambda X : T_2. \quad & \forall P : \mathbb{N} \to \Box.\ P\ (\uparrow_{\mathbb{N}} X.1) \to \\
  & \left(\forall n : \mathbb{N}, P\ n \to P \left(\uparrow_{\mathbb{N}} (X.2\ (\downarrow_{\mathbb{N}} n))\right)\right) \to \forall n : \mathbb{N}, P\ n.
\end{align*}

Automation devices dedicated to proof transfer thus typically consist
of a meta-program which attempts to compute type former $W$ and
 proof $w$ by induction on the structure of $V$, by composing
registered canonical pairs of related terms, and the corresponding
proofs. These tools differ by the nature of relations they can
accommodate, and by the class of type formers they are able to
synthesize. For instance, \emph{generalized
rewriting}~\cite{DBLP:journals/jfrea/Sozeau09}, which provides
essential support to formalizations based on
setoids~\cite{DBLP:journals/jfp/BartheCP03}, addresses the case of
homogeneous (and reflexive) relations, \ie when $T_1$ and $T_2$
coincide. The \CoqEAL{} library~\cite{DBLP:conf/cpp/CohenDM13}
provides another example of such transfer automation tool, geared
towards \emph{refinements}, typically from a proof-oriented
data-structure to a computation-oriented one. It is thus specialized
to heterogeneous, functional relations but restricted to closed,
quantifier-free type formers. We now discuss the few transfer methods
which can accommodate dependent types and heterogeneous relations.

\subsection{Type equivalences, univalence}\label{ssec:univ}

Let us first focus on the special case of types related by an
\emph{equivalence}, and start with a few standard definitions,
notations and lemmas. Omitted details can be found in the usual
references, like the Homotopy Type Theory book~\cite{hottbook}. Two
functions $f,g : A \rightarrow B$ are \emph{point-wise equal},
denoted $f \ptweq g$ when their values coincide on all arguments, that
is $f \ptweq g \triangleq \Pi a : A.\,f\ a = g\ a$. For any type $A$,
$id_A$ denotes $\lambda a : A.\,a$, the identity function on $A$, and
we write $id$ when the implicit type $A$ is not ambiguous.

\begin{definition}[Type isomorphism, type equivalence]\label{def:equiv}
A function $f : A \rightarrow B$ is an \emph{isomorphism}, denoted
$\isiso{f}$, if there exists a function $g : B \rightarrow A$ which
satisfies the section and retraction properties, \ie
$g$ is respectively a point-wise left and right inverse of
$f$. A function $f$ is an \emph{equivalence}, denoted $\isequiv{f}$, when it moreover enjoys a \emph{coherence property}, relating the proofs of the section and retraction properties and ensuring that $\isequiv{f}$ is proof-irrelevant.

Types $A$ and $B$ are \emph{equivalent}, denoted $A \tequiv B$, when
there is an equivalence $f : A \rightarrow B$:
\[ A \tequiv B \enspace \triangleq \enspace \Sigma f : A \rightarrow B.\,\isequiv{f}\]
\end{definition}

\begin{lemma}\label{lem:adj}
Any isomorphism $f : A \rightarrow B$ is also an equivalence.
\end{lemma}

The data of an equivalence $e : A \tequiv B$ thus include two
\emph{transport functions}, denoted respectively $\uparrow_e\ : A
\rightarrow B$ and $\downarrow_e\ : B \rightarrow A$. They can be
used for proof transfer from $A$ to $B$, using $\uparrow_e$ at
covariant occurrences, and $\downarrow_e$ at contravariant ones. The
\emph{univalence principle} asserts that equivalent types are
interchangeable, in the sense that all universes are univalent.

\begin{definition}[Univalent universe]
A universe $\mathcal{U}$ is univalent if for any two types $A$ and $B$ in
$\mathcal{U}$, the canonical map $A = B \rightarrow A \tequiv B$ is an
equivalence.
\end{definition}
In variants of \CComega, the \emph{univalence axiom} has no explicit computational content: it just postulates that all universes $\Box_i$ are univalent, as for instance in the \HoTT{} library for the \Coq{} interactive theorem prover~\cite{DBLP:conf/cpp/BauerGLSSS17}. Some more recent variants of
dependent type theory~\cite{cohen_et_al:LIPIcs.TYPES.2015.5,DBLP:journals/mscs/AngiuliBCHHL21} feature a built-in computational univalence
principle. They are used to implement experimental interactive theorem provers, such as
\CubicalAgda{}~\cite{DBLP:journals/pacmpl/VezzosiM019}. In both cases, the
univalence principle provides a powerful proof transfer principle from $\Type{}$
to $\Type{}$, as for any two types $A$ and $B$ such that $A \tequiv B$, and any
$P : \Type{} \rightarrow \Type{}$, we can obtain that $P\ A \tequiv P\ B$ as a
direct corollary of univalence. Concretely, $P\ B$ is obtained from $P\ A$ by
appropriately allocating the transfer functions provided by the equivalence
data, a transfer process typically useful in the context of proof
engineering~\cite{DBLP:conf/pldi/RingerPYLG21}.

Going back to our example from
\textsection~\ref{ssec:bigp}, transferring along an equivalence
$\mathbb{N} \tequiv \texttt{N}$ thus produces~$W''$ from $V'$.
Assuming univalence,
one may achieve the more informative
transport from $V'$ to $W'$, using a method called
\emph{univalent parametricity}~\cite{univparam2}, which we discuss in the next
section.

\subsection{Parametricity translations}\label{ssec:univparam}
Univalent parametricity strengthens the transfer principle provided by
the univalence axiom by combining it with parametricity.  In \CComega,
the essence of parametricity, which is to devise a relational
interpretation of types, can be turned into an actual syntactic
translation, as relations can themselves be modeled as
$\lambda$-terms in \CComega. The seminal work of Bernardy, Lasson \etal~\cite{DBLP:conf/fossacs/BernardyL11,DBLP:conf/csl/KellerL12,DBLP:journals/jfp/BernardyJP12} combine in what we refer to as the \emph{raw parametricity
translation}, which essentially defines inductively a logical relation
$\paramt{T}$ for any type $T$,
as described on Figure~\ref{fig:rawparam}.
This presentation uses the standard convention that $t'$ is the term
obtained from a term $t$ by replacing every variable $x$ in $t$ with a
fresh variable $x'$. A variable $x$ is translated into a variable
$x_R$, where $x_R$ is a fresh name. 
\begin{figure}[h]
  \begin{itemize}
  \item Context translation:
    \begin{align}\label{eqn:rawparam}
      \paramt{\emptyctxt} & = \emptyctxt\\
      \paramt{\Gamma, x : A} & = \paramt{\Gamma}, x : A, x' : A', x_R :
      \paramt{A}\ x\ x' \label{eq:rawparam-var}
    \end{align}
  \item Term translation:
    \begin{align}
      \paramt{\Type{i}} & = \lambda A\,A'.\,A \rightarrow A' \rightarrow \Type{i}\\
      \paramt{x}  & = x_R\\
      \paramt{A\ B} &= \paramt{A}\ B\ B'\ \paramt{B}\\
      \paramt{\lambda x : A.\,t} &= \lambda(x : A) (x' : A')(x_R :
        \paramt{A} \ x\ x').\,\paramt{t}\\
      \paramt{\Pi x: A.\,B} &= \lambda f\,f'.\,\Pi(x : A)(x' : A')(x_R :
        \paramt{A}\ x\ x').\,\paramt{B} (f\ x) (f'\ x')\label{eq:rawparampi}
    \end{align}
  \end{itemize}
  \caption{Raw parametricity translation for \CComega.}
\label{fig:rawparam}
\end{figure}
Parametricity follows from the associated fundamental theorem, also called abstraction theorem~\cite{DBLP:conf/ifip/Reynolds83}:

\begin{theorem}\label{thm:abs-raw}
  If $\tyju{\Gamma}{t}{T}$ then the following hold: $\tyju{\paramt{\Gamma}}{t}{T}$, $\tyju{\paramt{\Gamma}}{t'}{T'}$ and~$\tyju{\paramt{\Gamma}}{\paramt{t}}{\paramt{T}\ t\ t'}$.
\end{theorem}
\begin{proof} By structural induction on the typing judgment, see for instance~\cite{DBLP:conf/csl/KellerL12}.
\end{proof}
A key, albeit mundane ingredient of Theorem~\ref{thm:abs-raw}
is the fact that the rules of Figure~\ref{fig:rawparam} ensure that:
\begin{equation}\label{eq:rawparam}
\vdash \paramt{\Type{i}}\ :\ \paramt{\Type{i+1}}\ \Type{i}\ \Type{i}
\end{equation}
This translation precisely generates the statements expected from a
parametric type family or program. For instance, the translation of a
$\Pi$-type, given by Equation~\ref{eq:rawparampi}, is a type of
relations on functions that relate those producing related outputs
from related inputs. Concrete implementations of this translation are
available~\cite{DBLP:conf/csl/KellerL12,tassi:hal-01897468}; they generate and prove parametricity properties for type
families or for constants, improved induction schemes,~etc.

Univalent parametricity follows from the observation that the abstraction theorem still holds when restricting to relations
that are in fact (heterogeneous) equivalences. This however requires some care
in the translation of universes:
\begin{multline}\label{eqn:univparamdef}
  \paramt{\Type{i}}\ A\ B \enspace \triangleq \enspace \Sigma (R : A \rightarrow B
  \rightarrow \Type{i}) (e : A \tequiv B).\\ \Pi (a : A) (b : B).\,R\ a\ b
  \tequiv (a =\,\downarrow_e b)
\end{multline}
where $\paramt{\cdot}$ now refers to the \emph{univalent} parametricity
translation, replacing the notation introduced for the raw variant.
For any two types $A$ and $B$, $\paramt{\Type{i}}\ A\ B$ packages a relation $R$ and an equivalence $e$ such that $R$
is equivalent to the functional relation associated with $\downarrow_e$.
Crucially, assuming univalence,
$\paramt{\Type{i}}$ is equivalent to type equivalence, that is, for any two
types $A$ and $B$:
\[\paramt{\Type{i}}\ A\ B \tequiv (A \tequiv B).\]

This observation is actually an instance of a more general technique
available for constructing syntactic models of type
theory~\cite{DBLP:conf/cpp/BoulierPT17}, based on attaching extra
intensional specifications to negative type constructors. In these models, a
standard way to recover the abstraction theorem consists of
refining the translation into two variants, for any term $T : \Type{i}$,
that is also a type. The translation of such a $T$ as a \emph{term}, denoted
$\paramtm{T}$, is a dependent pair, which equips a relation
with the additional data prescribed by the interpretation
$\paramt{\Type{i}}$ of the universe. The translation $\paramt{T}$ of
$T$ as a \emph{type} is the relation itself, that is, the
projection of the dependent pair $\paramtm{T}$ onto its first
component, denoted $\relofparamtm{T}$. We refer to the original
article~\cite[Figure~4]{univparam2} for a complete description of
the translation.

We now state the abstraction theorem of the univalent parametricity translation~\cite{univparam2}, where
  $\vdash_u$ denotes a typing judgment of \CComega{} assuming univalence:
\begin{theorem}\label{thm:uabs}
  If $\tyju{\Gamma}{t}{T}$ then
  $\paramt{\Gamma}\vdash_u \paramtm{t}\ :\ \paramt{T}\ t\ t'$.
\end{theorem}
Note that proving the abstraction theorem~\ref{thm:uabs} involves
in particular proving that:
\begin{equation}\label{eq:uparam}
  \vdash_u \paramtm{\Type{i}}\ :\ \paramt{\Type{i+1}}\ \Type{i}\ \Type{i} \quad \textrm{and} \quad \relofparamtm{\Type{i}} \conv \paramt{\Type{i}}.
\end{equation}
The definition of relation
$\paramtm{\Type{i}}$ relies on the univalence principle in a crucial way, in order
to prove that the relation in the universe is equivalent to equality on
the universe, \ie to prove that:
\[\vdash_u \Pi A\,B : \Type{i}.\,\paramt{\Type{i}}\ A\ B\tequiv (A = B).\]
Importantly, this univalent parametricity translation can be seamlessly
extended so as to also make use of a global context of user-defined
equivalences.

Yet because of the interpretation of universes given by Equation~\ref{eqn:univparamdef}, univalent parametricity can only automate proof transfer based on \emph{type equivalences}. This is too strong a requirement in many cases, \eg to deduce properties of natural numbers from that of integers, or more generally for refinement relations. Even in the case of equivalent types, this restriction may be problematic, as Equation~\ref{eq:uparam} may incur unnecessary dependencies on the univalence axiom, as in Example~\ref{ex:elim}.


\section{Type equivalence in kit}\label{sec:disassembling-reassembling}

In this section, we propose (\textsection~\ref{ssec:disassembling}) an
 equivalent, modular presentation of type equivalence, phrased as a nested sigma type.
Then (\textsection~\ref{ssec:reassembling}), we carve a hierarchy of structures on relations out of this
dependent tuple,  selectively picking pieces. Last, we revisit (\textsection~\ref{ssec:popu}) parametricity
translations through the lens of this finer grained analysis of the relational interpretations of types.

\subsection{Disassembling type equivalence}\label{ssec:disassembling}
Let us first observe that the Definition~\ref{def:equiv}, of type
equivalence, is quite asymmetrical, although this fact is somehow swept
under the rug by the infix $A \tequiv B$ notation. First, the
data of an equivalence $e : A \tequiv B$ privileges the
left-to-right direction, as $\uparrow_e$ is directly accessible from
$e$ as its first projection, while accessing the right-to-left
transport requires an additional projection. Second, the statement of
the coherence property, which we eluded in
Definition~\ref{def:equiv}, is actually:
\[ \Pi a : A.\,\textsf{ap}_{\uparrow_e} (s\ a) =
r\ \circ\,\downarrow_e \]
where $\textsf{ap}_f (t)$ is the term $f\ u
= f\ v$, for any identity proof $t : u = v$.  This statement uses
proofs $s$ and $r$, respectively of the section and retraction
properties of $e$, but not in a symmetrical way, although swapping them
leads to an equivalent definition. This entanglement prevents
 tracing the respective roles of each direction of transport,
left-to-right or right-to-left, during the course of a given univalent
parametricity translation.  Exercise 4.2 in the HoTT
book~\cite{hottbook} however suggests a symmetrical 
definition of type equivalence, via functional relations.

\begin{definition}\label{def:hott-fun} A relation
  $R : A \rightarrow B \rightarrow \Type{i}$, is \emph{functional} when:
\[\Pi a : A.\,\iscontr{\Sigma b : B.\,R\ a\ b}\]
where for any type $T$, $\iscontr{T}$ is the standard contractibility predicate $\Sigma t : T.\,\Pi t' : T.\,t = t'$. This property is denoted $\isfun{R}$.
\end{definition}

We can now obtain an equivalent but symmetrical characterization of
type equivalence, as a functional relation whose symmetrization is also
functional.
\begin{lemma}\label{hott-exo}
  For any types $A, B : \Type{}$, type $A \tequiv B$ is equivalent to:
  \[\Sigma R : A \rightarrow B \rightarrow \Type{}.\,\isfun{R} \times \isfun{\rsym{R}}\]
where $\rsym{R} : B \rightarrow A \rightarrow \Type{}$ just swaps the arguments of relation $R : A \rightarrow B \rightarrow \Type{}$.
\end{lemma}
We sketch below a proof of this result, left as an exercise in~\cite{hottbook}. The essential argument is the following
characterization of functional relations:
\begin{lemma}\label{lem:isfun}
The type of functions is equivalent to the type of functional relations; \ie{} for any types $A, B\ :\ \Type{}$, we have $(A \rightarrow B) \, \tequiv \, \Sigma R : A \rightarrow B \rightarrow \Type{}.\,\isfun{R}$.
\end{lemma}
\begin{proof} The proof goes by chaining the following equivalences:
\begin{align*}
\left(\Sigma R : A \to B \to \Type{}.\,\isfun{R}\right) \enspace & \tequiv \enspace
   \left(A \to \Sigma P : B \to \Type{}.\,\iscontr{\Sigma b : B.\,P\ b}\right) \\
    & \enspace \tequiv (A \to B)
\end{align*}
\end{proof}

\begin{proof}[of Lemma~\ref{hott-exo}]
The proof goes by chaining the following equivalences, where the type of $f$ is
always $A \to B$ and the type of $R$ is $A \to B \to \Type{}$:
\begin{align*}
(A \tequiv B) & \enspace \tequiv \enspace \Sigma f : A \to B.\,\isequiv{f}
&& \text{by definition of }(A \tequiv B)
\\
& \enspace \tequiv \enspace
   \Sigma f.\,\Pi b : B.\,\iscontr{\Sigma a. f\ a = b}
&& \text{standard result in HoTT}
\\
& \enspace \tequiv \enspace
\Sigma f.\,\isfun{\lambda (b : B) (a : A).\,f\ a = b}
&& \text{by definition of } \isfun{\cdot}
\\
& \enspace \tequiv \enspace
\Sigma\left(\varphi : \Sigma R.\,\isfun{R}\right).\,\isfun{ \pi_1(\varphi)^{-1}}
&& \text{by Lemma~\ref{lem:isfun}} \\
& \enspace \tequiv \enspace
\Sigma R.\,\isfun{R} \times \isfun{R^{-1}}
&& \text{by associativity of } \Sigma.
\end{align*}
\end{proof}

However, the definition of type equivalence provided by
Lemma~\ref{hott-exo} does not expose explicitly the two
transfer functions in its data, although this computational content can
be extracted via first projections of contractibility proofs. In
fact, it is possible to devise a definition of type equivalence which
directly provides the two transport functions in its data, while
remaining symmetrical. This variant follows from 
an alternative characterization of functional relations.

\begin{definition}\label{def:umap} For any types $A, B : \Type{}$, a relation $R : A \rightarrow B \rightarrow
  \Type{}$, is a \emph{univalent map}, denoted $\isumap{R}$ when
  there exists a function $m : A \rightarrow B$ together with:
  \begin{flalign*}
    & g_1 : \Pi (a : A) (b : B).\,m\ a = b \to R\ a\ b \\
    \text{and} \enspace & g_2 : \Pi (a : A) (b : B).\,R\ a\ b \to m\ a = b \\
    \text{such that} \enspace &
   \Pi (a : A) (b : B).\,(g_1\ a\ b) \circ (g_2 \ a\ b) \ptweq id.
  \end{flalign*}
\end{definition}
Now comes the crux \trocqcode{theories/Uparam.v\#L73-L74}{lemma} of this section.
\begin{lemma} For any types $A, B : \Type{}$ and any relation $R : A \rightarrow B \rightarrow \Type{}$
\[\isfun{R} \tequiv \isumap{R}.\]
\end{lemma}
\begin{proof}
The proof goes by rewording the left hand side, in the following way:
\begin{align*}
  \Pi x. & \,\iscontr{R\ x} \\
    & \tequiv
    \Pi x.\,\Sigma (r : \Sigma y.\,R\ x\ y).\,\Pi(p : \Sigma y.\,R\ x\ y).\,r = p \\
    & \tequiv
    \Pi x.\,\Sigma y.\,\Sigma (r : R\ x\ y).\,\Pi(p : \Sigma y.\,R\ x\ y).\,(y, r) = p \\ & \tequiv
    \Sigma f.\,\Pi x.\,\Sigma (r : R\ x\ (f\ x)).\,\Pi(p : \Sigma y.\,R\ x\ y).\,(f\ x, r) = p \\
    & \tequiv
    \Sigma f.\,\Sigma (r : \Pi x.\,R\ x\ (f\ x)).\,\Pi x.\,\Pi(p : \Sigma y.\,
    R\ x\ y).\,(f\ x, r\ x) = p \\
    & \tequiv
    \Sigma f.\,\Sigma r.\,\Pi x.\,\Pi y.\,\Pi(p :
    R\ x\ y).\,(f\ x, r\ x) = (y, p) \\
    & \tequiv
    \Sigma f.\,\Sigma r.\,\Pi x.\,\Pi
    y.\,\Pi(p : R\ x\ y).\,\Sigma (e : f\ x = y).\,r\ x =_e p \\
    & \tequiv
    \Sigma f.\,\Sigma r.\,
    \Sigma (e : \Pi x.\,\Pi y.\,R\ x\ y \to f\ x = y).\,
      \Pi x.\,\Pi y.\,\Pi p.\,(r\ x) =_{e\,x\,y\,p} p
\end{align*}
After a suitable reorganization of the sigma types we are left to show that
\begin{align*}
  &\Sigma (r : \Pi x.\,\Pi y.\,f\ x = y \to R\ x\ y).\,(e\,x\,y) \circ (r\,x\,y) \ptweq id \\
  \tequiv \enspace & \Sigma (r : \Pi x.\,R\ x\ (f\ x)).\,\Pi x.\,\Pi y.\,\Pi p.\,r\ x =_{e\,x\,y\,p} p
\end{align*}
which proof we do not detail, referring the reader to the \trocqcode{theories/Uparam.v}{supplementary material}.
\end{proof}
As a direct corollary, we obtain a novel characterization of type
equivalence:
\begin{theorem}\label{thm:uparam-equiv}
  For any types $A, B : \Type{i}$, we have:
  \[(A \tequiv B) \tequiv \ParamRec{\top}\ A\ B\]
where the relation $\ParamRec{\top}\ A\ B$ is defined as:
  \[\Sigma R : A \rightarrow B \rightarrow \Type{i}.\,\isumap{R} \times \isumap{\rsym{R}}\]
\end{theorem}
The collection of data packed in a term of type $\ParamRec{\top}\ A\ B$ is now symmetrical, as the right-to-left
direction of the equivalence based on univalent maps can be obtained from the left-to-right by
flipping the relation and swapping the two functionality proofs. If the
$\eta$-rule for records is verified, symmetry is even \emph{definitionally}\ involutive.

\subsection{Reassembling type equivalence}\label{ssec:reassembling}
\label{ssec:pproof-box}

Definition~\ref{def:umap} of univalent maps and the resulting rephrasing of type equivalence
suggest introducing a hierarchy of structures for heterogeneous relations, which explains how
close a given relation is to type equivalence. In turn, this distance is described in terms of
structure available respectively on the left-to-right and right-to-left transport functions.

\begin{definition}\label{def:prels} For $n, k \in\{0, 1, 2_{\txta}, 2_{\txtb}, 3, 4\}$, and $\alpha=(n,k)$,
  relation $\ParamRec{\alpha}\ :\ \Type{} \to \Type{} \to \Type{}$, is defined as:
  \[\ParamRec{\alpha} \triangleq \lambda (A\ B\, :\, \Type{}).\Sigma(R\, :\, A \to B \to \Type{}).\ParamClass{\alpha}\ R\]
where the \emph{map class} $\ParamClass{\alpha}\ R$ itself unfolds to a pair type $(\ParamM{n}\ R) \times (\ParamM{k}\ \rsym{R})$,
with $\ParamM{i}$ defined as:\footnote{For the sake of readability, we omit implicit arguments, \eg{} although
$\ParamM{i}$ has type $\lambda (T_1\ T_2 \,:\, \Type{}).\,(T_1 \to T_2 \to \Type{}) \to \Type{}$, we write $\ParamM{n}\ R$ for $(\ParamM{n}\ A\ B\ R)$.}
\begin{align*}
\ParamM{0}\ R \triangleq &\ .\\
\ParamM{1}\ R \triangleq &\ (A \to B)\\
\ParamM{2_{\txta}}\ R \triangleq &\ \Sigma m : A \to B.\, G_{2_{\txta}}\ m\ R\quad \textrm{ with } G_{2_{\txta}}\ m\ R\ \triangleq \Pi a\,b.\,m\ a = b \to R\ a\ b \\
\ParamM{2_{\txtb}}\ R \triangleq &\ \Sigma m : A \to B.\,G_{2_{\txtb}}\ m\ R\quad \textrm{ with } G_{2_{\txtb}}\ m\ R\ \triangleq \Pi a\,b.\,R\ a\ b \to m\ a = b\\
\ParamM{3}\ R \triangleq &\ \Sigma m : A \to B.\,(G_{2_{\txta}}\ m\ R)\times (G_{2_{\txtb}}\ m\ R)\\
\ParamM{4}\ R \triangleq &\ \Sigma m : A \to B.\,\Sigma (g_1 : G_{2_{\txta}}\ m\ R).\,\Sigma (g_2 : G_{2_{\txtb}}\ m\ R).\,\Pi a\,b.\\
&\ (g_1\ a\ b) \circ (g_2\ a\ b) \ptweq id
\end{align*}
For any types $A$ and $B$, and any $r : \ParamRec{\alpha}\ A \ B$
we use notations $\Rel{r}$, $\map{r}$ and $\comap{r}$ to refer
respectively to the relation, map of type $A \to B$, map of type $B
\to A$, included in the data of $r$, for a suitable $\alpha$.
\end{definition}

\begin{definition}
  We denote $\ann$ the set $\{0, 1, 2_{\txta}, 2_{\txtb}, 3, 4\}^2$, used to index map classes in
  Definition~\ref{def:prels}. This set is partially ordered for the product order defined from the partial order
  $0 < 1 < 2_{*} < 3 < 4$ for $2_{*}$ either $2_\txta$ or $2_\txtb$, and with $2_\txta$ and $2_\txtb$ being incomparable.

\end{definition}
\begin{remark}
Relation $\ParamRec{(4,4)}$ of Definition~\ref{def:prels} coincides with the relation $\ParamRec{\top}$ introduced in
Theorem~\ref{thm:uparam-equiv}. Similarly, we denote $\ParamRec{\bot}$ the relation $\ParamRec{(0,0)}$.
\end{remark}

\begin{remark}


Definition~\ref{def:prels} is associated with the following dictionary. For $r$ of type:
\begin{itemize}
  \item $\ParamRec{(1,0)}\ A\ B$, $\map{r}$ is an arbitrary function $f : A \to B$;
  \item $\ParamRec{(4,0)}\ A\ B$, $\Rel{r}$ is a univalent map, in the sense of Definition~\ref{def:umap};
  \item $\ParamRec{(4,2_{\txta})}\ A\ B$, $\Rel{r}$ is the graph of a retraction (\ie{} a surjective univalent map with an explicit partial left inverse) of type $A \to B$;
  \item $\ParamRec{(4,2_{\txtb})}\ A\ B$,  $\Rel{r}$ is the graph of a section (\ie{} an injective univalent map with explicit partial right inverse) of type $A \to B$;
  \item $\ParamRec{(4,4)}$ , $r$ is an equivalence between $A$ and $B$;
  \item $\ParamRec{(3,3)}$, $r$ is an isomorphism between $A$ and $B$.
\end{itemize}
Observe that $\ParamRec{(n,m)}\ A\ B$ coincides, up to equivalence, with $\ParamRec{(m, n)}\ B\ A$. Other classes, while not corresponding to a meaningful mathematical definition, may arise in concrete runs of proof transfer: see also Section~\ref{sec:strat-param-rules} for explicit examples.

\end{remark}

The corresponding lattice to the collection of
$\ParamM{n}$ is implemented as a hierarchy of dependent tuples, more precisely,
of record types.



\subsection{Populating the hierarchy of relations}\label{ssec:popu}
We shall now revisit the parametricity translations of Section~\ref{ssec:univparam}. In particular,
combining Theorem~\ref{thm:uparam-equiv} with Equation~\ref{eq:uparam}, crux of the
abstraction theorem for univalent parametricity, ensures the existence of a term $p_{\Type{i}}$
such that:
\[\vdash_u p_{\Type{i}} \ :\ \ParamRec{\top}_{i+1}\ \Type{i}\ \Type{i} \quad \textrm{and} \quad \Rel{p_{\Type{i}}} \tequiv \ParamRec{\top}_{i}.\]
Otherwise said, relation $\ParamRec{\top}\, :\, \Type{} \to \Type{} \to \Type{}$ can be endowed with a
$\ParamRec{\top}$ structure, assuming univalence. Similarly, Equation~\ref{eq:rawparam}, for the
raw parametricity translation, can be read as the fact that
relation  $\ParamRec{\bot}$ on universes can be endowed with a
$\ParamRec{\bot}\ \Type{}\ \Type{}$ structure.

Now the hierarchy of structures introduced by Definition~\ref{def:prels} enables
a finer grained analysis of the possible relational interpretations of universes.
Not only would this put the raw and univalent parametricity translations under the same hood,
but it would also allow for generalizing parametricity to a larger class of relations.
For this purpose, we generalize the previous observation, on the key ingredient for
translating universes: for each $\alpha\in\ann$,
relation $\ParamRec{\alpha} :\ \Type{} \to \Type{} \to \Type{}$ may be endowed with
several structures from the lattice, and we need to study which ones, depending on~$\alpha$.
Otherwise said, we need to identify the pairs
$(\alpha, \beta)\in\ann^2$ for which it is possible to construct a term
$\pproof_\Box^{\alpha, \beta}$ such that:
\begin{equation}\label{eq:combi}
\vdash_u \pproof_\Box^{\alpha, \beta} \ :\ \ParamRec{\beta}\ \Type{}\ \Type{} \quad \textrm{and} \quad \Rel{p_{\Type{}}^{\alpha, \beta}} \conv \ParamRec{\alpha}
\end{equation}
Note that we aim here at a definitional equality between $\Rel{\pproof_\Box^{\alpha, \beta}}$ and
$\ParamRec{\alpha}$, rather than at an equivalence. It is easy to see that a term
$\pproof_\Box^{\alpha, \bot}$  exists for any $\alpha\in\ann$, as $\ParamRec{\bot}$ requires no structure
on the relation. On the other hand, it is not possible to construct a term $\pproof_\Box^{\bot, \top}$,
\ie{} to turn an arbitrary relation into a type equivalence.

\begin{definition}\label{def:dbox} We denote $\dep_\Box$ the following subset of $\ann^2$:
\[\dep_{\Box} = \{(\alpha, \beta)\in \ann^2\ |\ \alpha = \top \vee \beta \in \{0, 1, 2_{\txta}\}^2\}\]

\end{definition}

The \trocqcode{theories/Param_Type.v}{supplementary material} constructs
terms $\pproof_\Box^{\alpha, \beta}$ for every pair $(\alpha,\beta)\in\dep_\Box$, using a meta-program
to generate them from a minimal collection of manual definitions. In particular, assuming univalence, it is possible to
construct a term $p_{\Type{}}^{\top,\top}$, which can be seen as an analogue of the translation
\paramtm{\Type{}} of univalent parametricity. More generally, the provided terms
$\pproof_\Box^{\alpha, \beta}$ depend on univalence if and only if $\beta \notin \{0, 1, 2_{\txta}\}^2$.

The next natural question concerns the possible structures
$\ParamRec{\gamma}$ endowing the relational interpretation of a product type $\Pi x : A.\, B$, given relational interpretation for types $A$ and $B$
 respectively equipped with structures $\ParamRec{\alpha}$ and $\ParamRec{\beta}$.

Otherwise said, we need to identify the triples $(\alpha,\beta,\gamma)\in \ann^3$ for which it
is possible to construct a term $p_{\Pi}^{\gamma}$ such that the following statements both hold:
\[
\inferrule{\Gamma \vdash A_R\ :\ \ParamRec{\alpha}\ A\ A' \\
\Gamma,\ x : A,\ x' : A',\ x_R : A_R\ x\ x' \,\vdash\, B_R\ :\ \ParamRec{\beta}\ B\ B'
}{\Gamma\;\vdash\; p_{\Pi}^{\gamma}\ A_R\ B_R\ :\ \ParamRec{\gamma}\ (\Pi x : A.\, B)\ (\Pi x' : A'.\, B')}\]
\[\Rel{p_{\Pi}^{\gamma}\ A_R\ B_R} \conv \lambda f. \lambda f'. \Pi(x:A)(x':A')(x_R : \Rel{A_R}\ x\ x').\ \Rel{B_R}\ (f x)\ (f' x')\]

The corresponding collection of triples can actually be described as a function $\dep_\Pi : \ann \to \ann^2$, such that
$\dep_\Pi(\gamma) = (\alpha, \beta)$ provides the \emph{minimal} requirements on the structures
associated with $A$ and $B$, with respect to the partial order on $\ann^2$. The \trocqcode{theories/Param_forall.v}{supplementary material} provides a corresponding collection of terms $\pproof_\Pi^{\gamma}$
for each $\gamma\in\ann$, as well as all the associated weakenings. Once again, these definitions are
generated by a meta-program. Observe in particular that by symmetry,
$\pproof_\Pi^{(m, n)}$ can be obtained from $\pproof_\Pi^{(m, 0)}$ and $\pproof_\Pi^{(n, 0)}$
by swapping the latter and glueing it to the former. Therefore, the values of
$\pproof_\Pi^\gamma$ and $\dep_\Pi(\gamma)$ are completely determined by those of
$\pproof_\Pi^{(m, 0)}$ and $\dep_\Pi(m, 0)$.
In particular, for any $(m, n)\in \ann$:
\[ \dep_\Pi(m, n) = \left((m_A, n_A), (m_B, n_B)\right) \]
where $m_A, n_A, m_B, n_B\in \ann$ are such that
$\dep_\Pi(m, 0) = \left((0, n_A), (m_B, 0)\right)$ and $\dep_\Pi(n, 0) = \left((0, m_A), (n_B, 0)\right)$.
We sum up in Figure~\ref{fig:dep-pi} the values of $\dep_\Pi(m, 0)$.

\begin{figure}[h]
\begin{center}
\begin{minipage}{0.4\textwidth}
\begin{center}
\begin{tabular}{ |c|c|c| }
\hline
$m$ & $\dep_\Pi(m, 0)_1$ & $\dep_\Pi(m, 0)_2$ \\
\hline
$0$ & $(0,0)$ & $(0,0)$ \\
\hline
$1$ & $(0,2_{\txta})$ & $(1,0)$ \\
\hline
$2_{\txta}$ & $(0,4)$ & $(2_{\txta},0)$ \\
\hline
$2_{\txtb}$ & $(0,2_{\txta})$ & $(2_{\txtb},0)$ \\
\hline
$3$ & $(0,4)$ & $(3,0)$ \\
\hline
$4$ & $(0,4)$ & $(4,0)$ \\
\hline
\end{tabular}
\end{center}
\end{minipage}
\begin{minipage}{0.4\textwidth}
\begin{center}
\begin{tabular}{ |c|c|c| }
  \hline
  $m$ & $\dep_\rightarrow(m, 0)_1$ & $\dep_\rightarrow(m, 0)_2$ \\
  \hline
  $0$ & $(0,0)$ & $(0,0)$ \\
  \hline
  $1$ & \cellcolor{gray!50}$(0,1)$ & $(1,0)$ \\
  \hline
  $2_{\txta}$ & \cellcolor{gray!50} $(0,2_{\txtb})$ & $(2_{\txta},0)$ \\
  \hline
  $2_{\txtb}$ & $(0,2_{\txta})$ & $(2_{\txtb},0)$ \\
  \hline
  $3$ & \cellcolor{gray!50}$(0,3)$ & $(3,0)$ \\
  \hline
  $4$ & $(0,4)$ & $(4,0)$ \\
  \hline
  \end{tabular}
\end{center}
\end{minipage}
\end{center}
\caption{Minimal dependencies for product and arrow types}\label{fig:dep-pi}\label{fig:dep-arrow}
\end{figure}

Note that in the case of a
non-dependent product, constructing
$\pproof_\to^\gamma$ requires less structure on the domain $A$ of an arrow type
$A \to B$, which motivates the introduction of function $\dep_\to(\gamma)$. Using the combinator
for dependent products to interpret an arrow type, albeit correct, potentially pulls in
unnecessary structure (and axiom) requirements. The \trocqcode{theories/Param_arrow.v}{supplementary material} includes a construction of terms
$\pproof_\to^\gamma$ for any $\gamma\in\ann$.

The two tables in Fig.\ref{fig:dep-pi} show how
requirements on levels stay the same on the right hand side of both $\Pi$ and
$\to$, stay the same up to symmetries (exchange of variance and of $2_\txta$
and $2_\txtb$) on the left hand side of a $\to$ and increase on the left hand
side of a $\Pi$. This elegant arithmetic justifies our hierarchy of relations.

\section{A calculus for proof transfer}
\label{sec:strat-param-rules}
This section introduces \Trocq{}, a framework for proof transfer designed as a
generalization of parametricity translations, so as to allow for interpreting types as
instances of the structures introduced in Section~\ref{ssec:pproof-box}. We adopt a sequent style presentation,
which fits closely the type system of \CComega, while explaining
in a consistent way the transformations of terms and contexts. This choice of presentation
departs from the standard literature about parametricity in pure type systems. Yet, it
brings the presentation closer to actual implementations, whose necessary management of
parametricity contexts is swept under the rug by notational conventions
(\eg the primes of Section~\ref{ssec:univparam}).

For this purpose, we successively introduce four calculi, of increasing sophistication. We start
(\textsection~\ref{sec:raw-param}) with introducing this sequent style
presentation by rephrasing the raw parametricity translation, and the univalent parametricity
one (\textsection~\ref{sec:univ-param-triples}). We then introduce {\CCann} (\textsection~\ref{sec:CCann}),
a Calculus of Constructions with annotations on sorts and subtyping, before defining
(\textsection~\ref{sec:trocq}) the \Trocq{} calculus.

\subsection{Raw parametricity sequents}
\label{sec:raw-param}

We introduce \emph{parametricity contexts}, under the form of a list of triples
packaging two variables $x$ and $x'$ together with a third one  $x_R$. The latter $x_R$ is a \emph{witness} (a proof) that $x$ and $x'$ are related:
\[\Xi ::= \varepsilon\ |\ \Xi,\; \param{x}{x'}{x_R}\]
We write $(x, x', x_R) \in \Xi$ when $\Xi = \Xi',\; \param{x}{x'}{x_R},\; \Xi''$ for some $\Xi'$ and $\Xi''$.

We denote $\Var(\Xi)$
the sequence of variables related in a parametricity context $\Xi$,
with multiplicities:
\[\Var(\varepsilon)= \varepsilon \quad \quad \Var(\Xi,\; \param{x}{x'}{x_R}) =
\Var(\Xi), x, x', x_R\] A parametricity context $\Xi$ is
\emph{well-formed},  written $\Xi \vdash$, if the sequence $\Var(\Xi)$ is duplicate-free.
In this case, we use the notation $\Xi(x)=(x', x_R)$ as a synonym
of $(x, x', x_R) \in \Xi$.

A \emph{parametricity judgment} relates a parametricity context $\Xi$ and three
terms $M, M', M_R$ of $\CComega$. Parametricity judgments, denoted as:
\[ \Xi \vdash \param{M}{M'}{M_R},\]
are defined by rules of Figure~\ref{fig:rel-param} and read \emph{in context $\Xi$,
term $M$ translates to the term $M'$, because $M_R$}.

\begin{figure}[h]
\begin{mathpar}
\ninferrule
  {\ }
  {\Xi \vdash \param{\Type{i}}{\Type{i}}{\lambda (A\,B : \Type{i}).\,A \to B \to \Type{i}}}
  {ParamSort}

\ninferrule
  {(x, x', x_R) \in \Xi \\ \Xi \vdash }
  {\Xi \vdash \param{x}{x'}{x_R}}
  {ParamVar}

\ninferrule
  {\Xi \vdash \param{M}{M'}{M_R} \quad \Xi \vdash \param{N}{N'}{N_R}}
  {\Xi \vdash \param{M\ N}{M'\ N'}{M_R\ N\ N'\ N_R}}
  {ParamApp}

\ninferrule
  {\Xi, \param{x}{x'}{x_R} \vdash \param{M}{M'}{M_R}}
  {\Xi \vdash \param{\lambda x : A.\,M}{\lambda x' : A'.\,M'}{\lambda x\,x'\,x_R.\,M_R}}
  {ParamLam}

\ninferrule
  {x, x' \notin \Var(\Xi) \\\\
  \Xi \vdash \param{A}{A'}{A_R} \\
  \Xi, \param{x}{x'}{x_R} \vdash \param{B}{B'}{B_R}
  }
  {\Xi \vdash \param
    {\Pi x : A.\,B}{\Pi x' : A'.\,B'}
    {\lambda f\,g.\,\Pi x\,x'\,x_R.\,B_R\ (f\ x)\ (g\ x')}}
  {ParamPi}

\end{mathpar}
\caption{{\sc Param}: sequent-style binary parametricity translation}\label{fig:rel-param}
\end{figure}

\begin{lemma}\label{lem:param-rel-func}
  The relation associating a term $M$ with pairs $(M', M_R)$ such that
  $\Xi \vdash \param{M}{M'}{M_R}$ holds, with $\Xi$ a well-formed parametricity
context, is \emph{functional}. More precisely, for any well-formed parametricity context~$\Xi$:
\begin{align*}
  \forall M, M', N', M_R, N_R, \quad &
 \Xi \vdash \param{M}{M'}{M_R}\quad \wedge\quad \Xi \vdash \param{M}{N'}{N_R}\\ & \implies\ (M', M_R) = (N', N_R)
\end{align*}
\end{lemma}
\begin{proof} Immediate by induction on the syntax of $M$.\end{proof}

This presentation of parametricity thus provides an alternative definition of
translation $\paramt{\cdot}$ from Figure~\ref{fig:rawparam}, and accounts for the prime-based
notational convention used in the latter.

\begin{definition} A parametricity context $\Xi$ is \emph{admissible} for a
  well-formed typing context $\Gamma$, denoted $\adm{\Gamma}{\Xi}$, when $\Xi$
  and $\Gamma$ are well-formed as a parametricity context and $\Gamma$ provides
  coherent type annotations for all terms in $\Xi$, that is, for any variables
  $x, x', x_R$ such that $ \Xi(x) = (x', x_R)$, and for any terms $A'$ and
  $A_R$:
  \[ \Xi \vdash \param{\Gamma(x)}{A'}{A_R}\quad \implies\quad \Gamma(x') = A'
  \; \wedge \;
  \Gamma(x_R) \equiv A_R\ x\ x'\]
\end{definition}

We can now state and prove an abstraction theorem:
\begin{theorem}[Abstraction theorem]
\[
\inferrule{
  \Gamma \rhd \Xi \\
  \Gamma \vdash M : A \\
  \Xi \vdash \param{M}{M'}{M_R} \\
  \Xi \vdash  \param{A}{A'}{A_R}
  }{
  \Gamma \vdash M' : A' \\ \text{and} \\ \Gamma \vdash M_R : A_R\ M\ M'
  }
\]
\end{theorem}

\begin{proof}
By induction on the derivation of $\Xi \vdash \param{M}{M'}{M_R}$.
\end{proof}

\subsection{Univalent parametricity sequents}\label{sec:univ-param-triples}

We now propose in Figure~\ref{fig:univparam} a rephrased version of the univalent parametricity
translation~\cite{univparam2}, using the same sequent style and replacing the translation
of universes with the equivalent relation $\ParamRec{\top}$. Parametricity judgments are denoted:

\[
\Xi  \vdash_{u}   \param{M}{M'}{M_R}
\]
where $\Xi$ is a parametricity context and
$M$, $M'$, and $M_R$ are terms of \CComega{}. The $u$ index is a reminder that
typing judgments  $\Gamma \vdash_u M : A$ involved in the associated
abstraction theorem assume the univalence axiom.

\begin{figure}[h]
\begin{mathpar}
\ninferrule{\ }{\Xi \vdash_{u}  \param{\Box_i}{\Box_{i}}{
  \pproof_{\Box_i}^{\top,\top}}}{UParamSort}

\ninferrule
  {(x, x', x_R) \in \Xi \\ \Xi \vdash}
  {\Xi \vdash_{u} \param{x}{x'}{x_R}}
  {UParamVar}

\ninferrule
  {\Xi \vdash_{u} \param{M}{M'}{M_R} \\
    \Xi \vdash_{u} \param{N}{N'}{N_R}}
  {\Xi \vdash_{u} \param{M\ N}{M'\ N'}{M_R\ N\ N'\ N_R}}
  {UParamApp}

\ninferrule
  {\Xi \vdash_{u} \param{A}{A'}{A_R} \\
    \Xi, \param{x}{x'}{x_R} \vdash_{u} \param{M}{M'}{M_R}}
  {\Xi \vdash_{u} \param{\lambda x : A.\,M}{\lambda x' : A'.\,M'}{\lambda x\,x'\,x_R.\,M_R}}
  {UParamLam}

\ninferrule
  { \Xi \vdash_{u}  \param{A}{A'}{A_R} \\
    \Xi, \param{x}{x'}{x_R} \vdash_{u}  \param{B}{B'}{B_R}}
  {\Xi \vdash_{u}  \param{\Pi x: A.\,B}{\Pi x' : A'.\,B'}
  {\pproof_\Pi^\top\ A_R\ B_R}}
  {UParamPi}

\end{mathpar}
\caption{{\sc UParam}: univalent parametricity rules}
\label{fig:univparam}
\end{figure}

We can now rephrase the abstraction theorem for univalent parametricity.

\begin{theorem}[Univalent abstraction theorem]\label{thm:sequabs}
\[
\inferrule
{
\Gamma \rhd \Xi \\
\Gamma \vdash M : A \\
\Xi \vdash_{u} \param{M}{M'}{M_R} \\
\Xi \vdash_{u} \param{A}{A'}{A_R}
}
{\Gamma \vdash M' : A' \\ \text{and} \\
 \Xi \vdash_u M_R : \Rel{A_R}\ M\ M'}
\]
\end{theorem}
\begin{proof}
By induction on the derivation of $\Xi \vdash_{u} \param{M}{M'}{M_R}$.
\end{proof}
\begin{remark}
In Theorem~\ref{thm:sequabs}, $\Rel{A_R}$ is a term of type $A \to A' \to \Box$.
Indeed:
\[
\inferrule
{\Gamma \vdash A : \Box_i \\ \Xi \vdash_{u} \param{A}{A'}{A_R} \\
\Gamma \rhd \Xi}{
 \Gamma \vdash_{u} A_R : \Rel{\pproof_{\Box_i}^{\top,\top}}\ A\ A'}
\]
entails $A_R$ has type
\begin{align*}
  \Rel{\pproof_{\Box_i}^{\top,\top}}\ A\ A' \;\equiv\; &
    \ParamRec{\top}\ A\ A' \\
  \;\equiv\;  & \Sigma R : A \to A' \to \Box.\,\isumap{R} \times \isumap{R^{-1}}.
\end{align*}
\end{remark}



\subsection{Annotated type theory}
\label{sec:CCann}

We are now ready to generalize the relational interpretation of types provided by the
univalent parametricity translation, so as to allow for interpreting sorts with instances
of weaker structures than equivalence.
For this purpose, we introduce a variant {\CCann} of {\CComega} where each universe is
annotated with a label indicating the structure available on its relational interpretation.
Recall from Section~\ref{ssec:pproof-box} that we
have used annotations $\alpha \in \ann$ to identify the different structures of the lattice
disassembling type equivalence: these are the labels annotating sorts of  {\CCann}, so that
if $A$ has type $\Box^\alpha$, then the associated
relation $A_R$ has type $\ParamRec{\alpha}\ A\ A'$. The syntax of {\CCann} is thus:

\begin{gather*}
M, N, A, B \in \term{\CCann} ::= \Box_i^\alpha\ |\  x\ |\ M\ N \ |\ \lambda x : A.\,M\ |\ \Pi x : A.\,B\\
\alpha \in \ann = \{0, 1, 2_{\txta}, 2_{\txtb}, 3, 4\}^2 \qquad i \in \mathbb{N}
\end{gather*}

Before completing the actual formal definition of the \Trocq{} proof transfer framework, let us informally illustrate
how these annotations shall drive the interpretation of terms, and in particular, of a
dependent product $\Pi x : A.\,B$. In this case, before translating $B$,
three terms representing the bound variable $x$, its translation $x'$, and the parametricity
 witness $x_R$ are added to the context. The type of $x_R$ is $\Rel{A_R}\ x\ x'$ where $A_R$ is the parametricity witness relating $A$ to its translation $A'$.
The role of annotation $\alpha$ on the sort typing type $A$ is thus to
to govern the amount of information available in witness $x_R$, by
determining the type of $A_R$. This intent is reflected
in the typing rules of \CCann, which rely on the definition of the loci
$\dep_{\Box}$, $\dep_{\to}$ and $\dep_{\Pi}$, introduced in \textsection\ref{ssec:popu}.


Contexts are defined as usual, but typing terms in {\CCann} requires defining a \textit{subtyping} relation
$\subtype$, defined by the rules of Figure~\ref{fig:subtyping}. The typing rules of {\CCann} are
available in Figure~\ref{fig:typing-ann} and follow standard presentations~\cite{DBLP:journals/tcs/AspinallC01}. The $\equiv$ relation in the
(\textsc{SubConv}) rule is the \emph{conversion} relation, defined as the closure
of $\alpha$-equivalence and $\beta$-reduction.
The two types of judgment in \CCann{} are thus:
\[
  \Gamma \vdash_{+} A \subtype B \quad \text{and} \quad \Gamma \vdash_{+} M : A
\]
where $M, A$ and $B$ are terms in $\CCann$ and $\Gamma$ is a context in $\CCann$.

\begin{figure}[h]
\begin{mathpar}
\ninferrule
  {\Gamma \vdash_{+} A : K \\ \Gamma \vdash_{+} B : K \\ A \equiv B}
  {\Gamma \vdash_{+} A \subtype B}
  {SubConv}

\ninferrule
  {\alpha \geq \beta \\ i \leq j}
  {\Gamma \vdash_{+} \Box^{\alpha}_i \subtype \Box^{\beta}_{j}}
  {SubSort}

\ninferrule
  {\Gamma \vdash_{+} M'\ N : K \\ \Gamma \vdash_{+} M \subtype M'}
  {\Gamma \vdash_{+} M\ N \subtype M'\ N}
  {SubApp}

\ninferrule
  {\Gamma, x : A \vdash_{+} M \subtype M'}
  {\Gamma \vdash_{+} \lambda x : A.\,M \subtype \lambda x : A.\,M'}
  {SubLam}

\ninferrule
  {\Gamma \vdash_{+} \Pi x : A.\,B : \Box_i \\ \Gamma \vdash_{+} A' \subtype A \\ \Gamma, x : A' \vdash_{+} B \subtype B'}
  {\Gamma \vdash_{+} \Pi x : A.\,B \subtype \Pi x : A'.\,B'}
  {SubPi}

K ::= \Box_i\ |\ \Pi x : A.\,K
\end{mathpar}
\caption{Subtyping rules for \CCann}\label{fig:subtyping}
\end{figure}

\begin{figure}
\begin{mathpar}

\ninferrule
  {\Gamma \vdash_{+} M : A \\ \Gamma \vdash_{+} A \subtype B}
  {\Gamma \vdash_{+} M : B}
  {Conv$^+$}

\ninferrule
  { (\alpha, \beta) \in \dep_\Box }
  {\Gamma \vdash_{+} \Box_i^\alpha : \Box_{i+1}^{\beta}}
  {Sort$^+$}

\ninferrule
  {(x, A) \in \Gamma \\ \Gamma \vdash_{+}}
  {\Gamma \vdash_{+} x : A}
  {Var$^+$}

\ninferrule
    {\Gamma \vdash_{+} A : \Box_i \\ x \notin \Var(\Gamma)}
    {\Gamma,\ x : A \vdash_{+}}
    {Context$^+$}

\ninferrule
  {\Gamma \vdash_{+} M : \Pi x : A.\,B \\ \Gamma \vdash_{+} N : A}
  {\Gamma \vdash_{+} M\ N : B[x := N]}
  {App$^+$}

\ninferrule
  {\Gamma, x : A \vdash_{+} M : B}
  {\Gamma \vdash_{+} \lambda x : A.\,M : \Pi x : A.\,B}
  {Lam$^+$}

\ninferrule
  {\Gamma \vdash_{+} A : \Box_i^\alpha \\ \Gamma \vdash_{+} B : \Box_i^\beta
    \\ \dep_\to(\gamma) = (\alpha, \beta)}
  {\Gamma \vdash_{+} A \to B : \Box_i^\gamma}
  {Arrow$^+$}

\ninferrule
  {\Gamma \vdash_{+} A : \Box_i^\alpha \\ \Gamma, x : A \vdash_{+} B : \Box_i^\beta
    \\ \dep_\Pi(\gamma) = (\alpha, \beta)}
  {\Gamma \vdash_{+} \Pi x : A.\,B : \Box_i^\gamma}
  {Pi$^+$}

\end{mathpar}
\caption{Typing rules for \CCann}\label{fig:typing-ann}
\end{figure}

Due to space constraints, we omit the direct proof that \CCann is a conservative extension over \CComega. It goes by defining an erasure
function for terms $\eraseAnn{\,\cdot\,} : \term{\CCann} \to \term{\CComega}$
and the associated erasure function for contexts.

\subsection{The {\Trocq} calculus}\label{sec:trocq}
The final stage of the announced generalization consists in building an analogue
to the parametricity translations available in pure type systems, but for the
annotated type theory of \textsection~\ref{sec:CCann}. This analogue is geared
towards proof transfer, as discussed in \textsection~\ref{ssec:bigp}, and
therefore designed to \emph{synthesize} the output of the translation
from its input, rather than to \emph{check} that certain pairs of terms
are in relation. However, splitting up the interpretation of universes into a
lattice of possible relation structures means that the source term of the
translation is not enough to characterize the desired output: the translation
needs to be informed with some extra information about the expected outcome of
the translation. In the \Trocq{} calculus, this extra information is a
type of $\CCann$.

We thus define {\Trocq} \emph{contexts} as lists of quadruples:
\[
\Delta ::= \varepsilon\ |\ \Delta,\; \paramq{x}{A}{x'}{x_R}
\quad \text{where}\; A \in \mathcal{T}_{\CCann},
\]
and introduce a conversion function $\gamma$ from {\sc Trocq} contexts to \CCann
contexts:
\begin{align*}
\gamma(\varepsilon) \quad = \quad & \varepsilon \\
\gamma(\Delta, \; \paramq{x}{A}{x'}{x_R}) \quad = \quad & \gamma(\Delta),\; x : A
\end{align*}

Now, a {\Trocq} judgment is a 4-ary relation of the form $\Delta \vdash_t
\paramq{M}{A}{M'}{M_R}$, which is read \emph{in context $\Delta$, term~$M$ of
annotated type $A$ translates to term~$M'$, because $M_R$} and $M_R$ is called a
parametricity witness. {\Trocq} judgments
are defined by the rules of Figure~\ref{fig:trocq}. This definition involves a
weakening function for parametricity witnesses, defined as follows.


\begin{definition}
For all $p, q  \in \{0, 1, 2_{\txta}, 2_{\txtb}, 3, 4\}$, such that $p \geq q$, we define
 map $\downarrow^p_q : \ParamM{p} \to \ParamM{q}$, which forgets the fields from class $\ParamM{p}$ that are not in $\ParamM{q}$.

For all $\alpha, \beta \in \ann$, such that $\alpha \geq \beta$,  function $\downdownarrows^\alpha_\beta : \ParamRec{\alpha}\ A\ B \to \ParamRec{\beta}\ A\ B$ is defined by:
\[
  \downdownarrows^{(m,n)}_{(p,q)}\ \langle R, M^\rightarrow, M^\leftarrow \rangle := \langle R, \downarrow^m_p M^\rightarrow, \downarrow^n_q M^\leftarrow \rangle.
\]

The weakening function on parametricity witnesses is defined on Figure~\ref{fig:weak} by extending
function $\downdownarrows^\alpha_\beta$ to all relevant pairs of types of \CCann, \ie  $\Downarrow^T_U$ is defined
for $T, U \in \mathcal{T}_{\CCann}$ as soon as $T \subtype U$.
\end{definition}

\begin{figure}[h]
  \begin{mathpar}
  \ninferrule
    {(\alpha, \beta) \in \dep_\Box}
    {\Delta \vdash_t \paramq{\Box_i^{\alpha}}{\Box_{i+1}^{\beta}}{\Box_i^{\alpha}}{\pproof_{\Box_i}^{\alpha, \beta}}}
    {TrocqSort}
  
  \ninferrule
    { (x, A, x', x_R) \in \Delta \\ \gamma(\Delta) \vdash_{+}}
    {\Delta \vdash_t \paramq{x}{A}{x'}{x_R}}
    {TrocqVar}
  
    \ninferrule
    {\Delta \vdash_t \paramq{M}{\Pi x : A.\,B}{M'}{M_R} \\
      \Delta \vdash_t \paramq{N}{A}{N'}{N_R}}
    {\Delta \vdash_t \paramq{M\ N}{B[x := N]}{M'\ N'}{M_R\ N\ N'\ N_R}}
    {TrocqApp}
  
  \ninferrule
    {\Delta \vdash_t \paramq{A}{\Box_i^\alpha}{A'}{A_R} \\\\
      \Delta, \paramq{x}{A}{x'}{x_R} \vdash_t \paramq{M}{B}{M'}{M_R}}
    {\Delta \vdash_t \paramq{\lambda x : A.\,M}{\Pi x : A.\,B}{\lambda x' : A'.\,M'}
    {\lambda x\,x'\,x_R.\,M_R}}
    {TrocqLam}
  
  \ninferrule{
    (\alpha, \beta) = \dep_\rightarrow(\delta)\\\\
       \Delta \vdash_t \paramq{A}{\Box_i^{\alpha}}{A'}{A_R} \\
      \Delta \vdash_t \paramq{B}{\Box_i^{\beta}}{B'}{B_R}\\
       }
    {\Delta \vdash_t \paramq{A \rightarrow B}{\Box_i^\delta}{A' \rightarrow B'}
    {\pproof_\rightarrow^{\delta}\ A_R\ B_R}}
    {TrocqArrow}
  
  \ninferrule
    {
      (\alpha, \beta) = \dep_\Pi(\delta) \\
      \Delta \vdash_t \paramq{A}{\Box_i^{\alpha}}{A'}{A_R} \\\\
      \Delta, \paramq{x}{A}{x'}{x_R} \vdash_t \paramq{B}{\Box_i^{\beta}}{B'}{B_R}
     }
    {\Delta \vdash_t \paramq{\Pi x : A.\,B}{\Box_i^\delta}{\Pi x' : A'.\,B'}
    {\pproof_\Pi^{\delta}\ A_R\ B_R}}
    {TrocqPi}
  
  \ninferrule
    {\Delta \vdash_t \paramq{M}{A}{M'}{M_R} \\ \gamma(\Delta) \vdash_{+} A \subtype B}
    {\Delta \vdash_t \paramq{M}{B}{M'}{\Downarrow^A_{B} M_R}}
    {TrocqConv}
  \end{mathpar}
  \caption{{\sc Trocq} rules}\label{fig:trocq}
  \end{figure}

\begin{figure}
\begin{mathpar}
\big\Downarrow^{\Box_i^\alpha}_{\Box_i^{\alpha'}}\ t_R :=\
  \downdownarrows^\alpha_{\alpha'} t_R

\big\Downarrow^{A\ M}_{A'\ M'}\ N_R := \big\Downarrow^A_{A'}\ M\ M'\ N_R

\big\Downarrow^{\lambda x : A.\,B}_{\lambda x : A'.\,B'}\ M\ M'\ N_R :=
  \big\Downarrow^{B[x := M]}_{B'[x := M']}\ N_R

\big\Downarrow^{\Pi x : A.\,B}_{\Pi x : A'.\,B'}\ M_R :=
  \lambda x\ x'\ x_R.\ \big\Downarrow^B_{B'}\
  \big(M_R\ x\ x'\ (\big\Downarrow^{A'}_A\ x_R)\big)

\big\Downarrow^A_{A'}\ M_R := M_R
\end{mathpar}
\captionof{figure}{Weakening of parametricity witnesses}
\label{fig:weak}
\end{figure}

 An abstraction theorem relates \Trocq{} judgments and typing in \CCann.

\begin{theorem} [{\sc Trocq} abstraction theorem]\label{thm:trocq-abstraction}
\[
\inferrule
{ \gamma(\Delta) \vdash_{+} \\
  \gamma(\Delta) \vdash_{+} M : A\\
\Delta \vdash_t \paramq{M}{A}{M'}{M_R} \\
\Delta \vdash_t \paramq{A}{\Box_i^\alpha}{A'}{A_R}}
{\gamma(\Delta) \vdash_{+} M' : A' \\ \text{and} \\
 \gamma(\Delta) \vdash_{+} M_R : \Rel{A_R}\ M\ M'}
\]
\end{theorem}
\begin{proof}
  By induction on derivation $\Delta \vdash_t \paramq{M}{A}{M'}{M_R}$.
\end{proof}
Note that type $A$ in the typing hypothesis $\gamma(\Delta) \vdash_{+} M : A$ of the
abstraction theorem is exactly the extra information passed to the
translation. The latter can thus also be seen as an inference algorithm,
which infers annotations for the output of the translation from that of
the input.

\begin{remark}\label{rmk:infer}
Since by definition of $\pproof_\Box^{\alpha,\beta}$ (Equation~\ref{eq:combi}), we have
  $\vdash_t \paramq{\Box^\alpha}{\Box^\beta}{\Box^\alpha}{\pproof_\Box^{\alpha,\beta}}$,
  by applying Theorem~\ref{thm:trocq-abstraction} with $\gamma(\Delta) \vdash_{+} A : \Box^\alpha$, we get:
\[
\inferrule
{\gamma(\Delta) \vdash_{+} A: \Box^\alpha \\
\Delta \vdash_t \paramq{A}{\Box^\alpha}{A'}{A_R} \\
}{
 \gamma(\Delta) \vdash_{+} A_R : \Rel{\pproof_\Box^{\alpha,\beta}}\ A\ A'
}.
\]
Now by the same definition, for any $\beta\in\ann$, $\Rel{\pproof_\Box^{\alpha,\beta}} =
\ParamRec{\alpha}$, hence $\gamma(\Delta) \vdash A_R : \ParamRec{\alpha}\ A\
A'$, as expected by the type annotation $A : \Box^\alpha$ in the premise of the rule.

\end{remark}

\begin{remark}\label{rmk:infer-univ}
By applying the Remark~\ref{rmk:infer} with $\vdash_{+} \Box^\alpha :
\Box^\beta$, we indeed obtain that~\hbox{$\vdash_{+} \pproof_\Box^{\alpha,\beta} : \ParamRec{\beta}\
\Box^\alpha\  \Box^\alpha$} as expected, provided that $(\alpha, \beta) \in
\dep_\Box$.
\end{remark}

\subsection{Constants}
\label{sec:constants}

Concrete applications require extending \Trocq{} with constants.
Constants are similar to variables,
except that they are stored in a global
context instead of a typing context. A crucial difference though is that a constant may be assigned several different annotated
types in \CCann.

Consider for example, a constant $\mono{list}$, standing for the type of polymorphic lists.
As $\mono{list}\ A$ is the type of lists with elements of type $A$, it
can be annotated with type $\Box^\alpha \to \Box^\alpha$ for
any $\alpha \in \ann$.

Every constant declared in the global environment has an associated collection of possible
annotated types $ \constantTypes{c} \subset \mathcal{T}_{\CCann}$. We require that
all the annotated types of a same constant share the same erasure
in \CComega, \ie $\forall c, \forall A, \forall B,\ A, B \in \constantTypes{c} \Rightarrow \eraseAnn{A} = \eraseAnn{B}$.
For example, $\constantTypes{\mono{list}} =
  \left\lbrace\Box^\alpha \to \Box^\alpha\ \middle|\ \alpha \in \ann \right\rbrace.$

In addition, we provide translations $\dep_c(A)$ for each possible annotated type $A$ of each constant~$c$ in
the global context. For example,
$\dep_{\mono{list}}(\Box^{(1,0)} \to \Box^{(1,0)})$ is equal to
$\left(\mono{list},\ 
\lambda A\,A'\,A_R.\,\left(\mono{List.All2}\ A_R,\; \mono{List.map}\ \left(\mono{map}\ A_R\right)\right)\right)$,
 where relation \texttt{List.All2 $A_R$} relates lists of the same length, whose elements are pair-wise related via $A_R$, $\mono{List.map}$ is the usual map function on lists and $\mono{map}\ A_R : A \to A'$ extracts the \emph{map} projection of the record
$A_R$ of type $\ParamRec{(1,0)}\ A\ A' \equiv \Sigma R.\,A \to A'$. Part of
these translations can be generated automatically by weakening.

We describe in Figure~\ref{fig:constant-rules} the additional rules for
constants in \CCann and {\sc Trocq}. Note that for an input
term featuring constants, an unfortunate choice of annotation may
lead to a stuck translation.


\begin{figure}[h]
\begin{mathpar}
\centering
\ninferrule
  {c \in \constants{} \\ A \in \constantTypes{c}}
  {\Gamma \vdash c : A}
  {Const$^+$}

\ninferrule
  {\dep_c(A) = (c', c_R)}
  {\Delta \vdash \paramq{c}{A}{c'}{c_R}}
  {TrocqConst}

\end{mathpar}
\caption{Additional constant rules for \CCann and {\Trocq}}\label{fig:constant-rules}
\end{figure}

We describe in Figure~\ref{fig:constant-rules} the additional rules for
constants in \CCann and {\sc Trocq}. Note that for an input
term featuring constants, an unfortunate choice of annotation may
lead a stuck translation.

\vspace{-2ex}
\section{Implementation and applications}\label{sec:applications}



The \Trocq{} plugin~\cite{trocq_zenodo} turns the material presented in Section~\ref{sec:strat-param-rules} into an actual tactic, called \coq{trocq}, for automating proof transfer in \Coq{}. This tactic synthesizes a new goal from the current one, as well as a proof that the former implies the latter. User-defined relations between constants, registered via specific vernacular commands, inform this synthesis. The core of the plugin implements each rule of the \Trocq{}
 calculus in the \Elpi{} meta-programming language~\cite{DBLP:conf/lpar/DunchevGCT15,tassi:hal-01897468}, on top of \Coq{} libraries formalizing the
 contents of Section~\ref{sec:disassembling-reassembling}.
 In the logic programming paradigm of {\Elpi}, each rule of Figure~\ref{fig:trocq} translates
 gracefully into a corresponding \LambdaProlog{} predicate, making the corresponding source code very close
 to the presentation of \textsection\ref{sec:trocq}.
 However, the \Trocq{} plugin also implements a much less straightforward annotation inference algorithm, so as to hide the management of sort annotations to  {\Coq} users completely. This section illustrates the usage of the \coq{trocq} tactic on various concrete examples.

\subsection{Isomorphisms}\label{ssec:ex-iso}
\paragraph{Bitvectors \trocqcode{examples/Vector_tuple.v\#L284}{(code)}. }
Here are two possible encodings of bitvectors in \Coq{}:
\begin{minted}{coq}
bounded_nat (k : nat) := {n : nat & n < pow 2 k}. (* n < 2^k *)
bitvector (k : nat) := Vector.t Bool k. (* size k vectors of booleans *)
\end{minted}
We can prove that these representations are equivalent by combining two proofs by transitivity: the proof that \coq{bounded_nat k} is related to \coq{bitvector k} for a given \coq{k}, and the proof that \coq{Vector.t} is related to itself. We also make use of the equivalence relations \coq{natR} and \coq{boolR}, which respectively relate type \coq{nat} and \coq{Bool} with themselves:
\begin{minted}{coq}
Rk : /$\forall$/ (k : nat), Param44.Rel (bounded_nat k) (bitvector k)
vecR : /$\forall$/ (A A' : Type) (AR : Param44.Rel A A') (k k' : nat)
  (kR : natR k k'), Param44.Rel (Vector.t A k) (Vector.t A' k')
(* equivalence between types (bounded_nat k) and (bitvector k') *)
bvR : /$\forall$/ (k k' : nat) (kR : natR k k'),
  Param44.Rel (bounded_nat k) (bitvector k')
(* informing Trocq with these equivalences *)
Trocq Use vecR natR boolR bvR.
\end{minted}

Now, suppose we would like to transfer the following result from the bounded natural numbers to the vector-based encoding:
\begin{minted}{coq}
/$\forall$/ (k : nat) (v : bounded_nat k) (i : nat) (b : Bool), i < k ->
  get (set v i b) i = b
\end{minted}
As this goal involves \coq{get} and \coq{set} operations on bitvectors, and the order and equality relations on type \coq{nat}, we inform \Trocq{} with the associated operations \coq{getv} and \coq{setv} on the vector encoding.
E.g., for \coq{get} and \coq{getv}, we prove:
\begin{minted}{coq}
getR : /$\forall$/ (k k' : nat) (kR : natR k k')
  (v : bounded_nat k) (v' : bitvector k') (vR : bvR k k' kR v v')
  (n n' : nat) (nR : natR n n'), boolR (get v n) (getv v' n')
\end{minted}
We can now use proof transfer from bitvectors to bounded natural numbers:
\begin{minted}{coq}
Trocq Use eqR ltR. (* where eq and lt are translated to themselves *)
Trocq Use getR setR.

Lemma setBitGetSame : /$\forall$/ (k : nat) (v : bitvector k),
/$\forall$/ (i : nat) (b : Bool), i < k -> getv (setv v i b) i = b.
Proof. trocq. exact setBitGetSame'. (* same lemma, on bitvector *) Qed.
\end{minted}
\paragraph{Induction principle on integers. \trocqcode{examples/peano_bin_nat.v}{(code)}.}
Recall that the problem from Example~\ref{ex:elim} is to obtain the following elimination scheme, from that available on type~$\mathbb{N}$:

\begin{minted}{coq}
N_ind : /$\forall$/ P : N /$\rightarrow$/ /$\square$/, P O/$_{\texttt{N}}$/ /$\rightarrow$/ (/$\forall$/ n : N, P n /$\rightarrow$/ P (S/$_{\texttt{N}}$/ n)) /$\rightarrow$/ /$\forall$/ n : N, P n
\end{minted}

We first inform \Trocq{} that \coq{N} and \coqescape{/$\mathbb{N}$/} are isomorphic, by providing proofs that the two conversions $\tonat\ :\
\texttt{N}\ \to\ \mathbb{N}$ and $\ofnat\ :\ \mathbb{N}\ \to\ \texttt{N}$
are mutual inverses. Using lemma \coq{Iso.toParam}, we can deduce an equivalence \coq{Param44.Rel N
/$\mathbb{N}$/}, \ie $\ParamRec{(4,4)}\ {\tt N}\ \mathbb{N}$. We also prove and register that zeros and successors are related:
\begin{minted}{coq}
Definition N/$_{\textrm{R}}$/ : Param44.Rel N /$\mathbb{N}$/ := ...
Lemma O/$_{\textrm{R}}$/ : rel N/$_{\textrm{R}}$/ O/$_{\textrm{N}}$/ O/$_{\mathbb{N}}$/.
Lemma S/$_{\textrm{R}}$/ : /$\forall$/ m n, rel N/$_{\textrm{R}}$/ m n /$\to$/ rel N/$_{\textrm{R}}$/ (S/$_{\textrm{N}}$/ m) (S/$_{\mathbb{N}}$/ n).
Trocq Use N/$_{\textrm{R}}$/ O/$_{\textrm{R}}$/ S/$_{\textrm{R}}$/.
\end{minted}
\Trocq{} is now able to generate, prove and apply the desired implication:
\begin{minted}{coq}
Lemma N_ind : /$\forall$/ P : N /$\rightarrow$/ /$\square$/, P O/$_{\texttt{N}}$/ /$\rightarrow$/ (/$\forall$/ n : N, P n /$\rightarrow$/ P (S/$_{\texttt{N}}$/ n)) /$\rightarrow$/
  /$\forall$/ n : N, P n.
Proof.
  trocq. (* in the goal, N, O$_{\texttt{N}}$, S$_{\texttt{N}}$ have been replaced by $\mathbb{N}$, O$_{\mathbb{N}}$, S$_{\mathbb{N}}$ *)
  exact nat_rect.
Qed.
\end{minted}
Inspecting this proof confirms that only information up to level $(2_\txta, 3)$ has been extracted from the equivalence proof \coqescape{N/$_{\textrm{R}}$/}. It is thus possible to run the exact proof transfer, but with a weaker relation, as illustrated in the \trocqcode{examples/nat_ind.v}{code} for an abstract type $I$ with a zero and a successor constants, and a retraction $\mathbb{N} \to I$.

\subsection{Sections, retractions}\label{ssec:ex-sec-ret}

\paragraph{Modular arithmetic \trocqcode{examples/int_to_Zp.v}{(code)}.}

A typical application of modular arithmetic is to show that some statement on $\mathbb{Z}$ can be reduced to statments on $\mathbb{Z}/p\mathbb{Z}$
Let us show how \Trocq{} can synthesize and prove the following implication:
\begin{minted}{coq}
Lemma IntRedModZp : (forall (m n p : Zmodp), (m = n * n)%Zmodp -> m = 0)
  -> forall (m n p : int), (m = n * n)%int -> (m == 0)%int.
Proof. intro Hyp. trocq; simpl. now apply Hyp. Qed.
\end{minted}
where scope \coq{/$\%$/Zmodp} is for the usual product and zero on type \coq{Zmodp}, for $\mathbb{Z}/p\mathbb{Z}$, scope \coq{/$\%$/int} for those on type \coq{int}, for $\mathbb{Z}$, and \coq{==} is an equality test modulo $p$ on type \coq{int}. Observe that the implication deduces a lemma on $\mathbb{Z}$ \emph{from} its modular analogues. Type \coq{Zmodp} and \coq{int} are obviously not equivalent, but a \emph{retraction} is actually enough for this proof transfer. We have:
\begin{minted}{coq}
modp : int -> Zmodp
reprp : Zmodp -> int
reprpK : /$\forall$/ (x : Zmodp), modp (reprp x) = x
Rp : Param42a.Rel int Zmodp
\end{minted}
where \coq{Rp}, (a proof that $\ParamRec{(4, 2_\txta)}\ \ \mathbb{Z}\ \ \mathbb{Z}/p\mathbb{Z}$), is obtained from \coq{reprpK} via lemma \coq{SplitSurj.toParam}. Proving lemma \coq{IntRedModZp} by \coq{trocq} now just requires relating the respective zeroes, multiplications, and equalities of the two types: 
\begin{minted}{coq}
R0 : Rp 0%int 0%Zmodp.
Rmul : /$\forall$/ (m : int) (x : Zmodp) (xR : Rp m x)
  (n : int) (y : Zmodp) (yR : Rp n y), Rp (m * n)%int (x * y)%Zmodp.
Reqmodp : /$\forall$/ (m : int) (x : Zmodp), Rp m x ->
  /$\forall$/ (n : int) (y : Zmodp), Rp n y -> Param01.Rel (m == n) (x = y).
Trocq Use Rp Rmul R0 Reqmodp. (* informing Trocq with these relations *)
\end{minted}
where \coq{Param01.Rel P Q} (\coq{Param01.Rel} is the \Coq{} name for $\ParamRec{(0, 1)}$) is \coq{Q -> P}. Note that by definition of the relation given by \coq{Rp}, lemma \coq{Rmul} amounts to:
\begin{minted}{coq}
  /$\forall$/ (m n : int), modp (m * n)%int = (modp m * modp n)%Zmodp.
\end{minted}


\paragraph{Summable sequences. \trocqcode{examples/summable.v}{(code)}.} Example~\ref{ex:sums} involves two instances of subtypes: type $\xnnR$ extends a type $\nnR$ of positive real numbers with an abstract element and type \coq{summable} is for provably convergent sequences of positive real numbers:
\begin{minted}{coq}
Inductive /$\xnnR$/ : Type := Fin : /$\nnR$/ /$\to$/ /$\xnnR$/ | Inf : /$\xnnR$/. 
Definition seq/$_\nnR$/ := nat /$\to$/ /$\nnR$/. Definition seq/$_\xnnR$/ := nat /$\to$/ /$\xnnR$/.
Record summable := {to_seq :> seq/$_\nnR$/; _ : isSummable to_seq}.
\end{minted}
Type $\xnnR$ and $\nnR$ are related at level $(4,2_\txtb)$: \eg \coq{truncate : /$\xnnR$/ -> /$\nnR$/} provides a partial inverse to the \coq{Fin} injection by sending the extra \coq{Inf} to zero.
Types \coq{summable} and \coq{seq/$_\xnnR$/} are also related at level $(4,2_\txtb)$, via the relation:
\begin{minted}{coq}
Definition Rrseq (u : summable) (v : seq/$_\xnnR$/) : Type := seq_extend u = v.
\end{minted}
where \coq{seq_extend} transforms a summable sequence into a sequence of extended positive reals in the obvious way.
Now \coq{/$\Sigma_\xnnR$/ u : /$\xnnR$/} is the sum of a sequence \coq{u : seq/$_\xnnR$/} of extended positive reals, and we also define the sum of a sequence of positive reals, as a positive real, again by defaulting infinite sums to zero. For the purpose of the example, we only do so for summable sequences:
\begin{minted}{coq}
Definition /$\Sigma_\nnR$/ (u : summable) : /$\nnR$/ := truncate (/$\Sigma_\xnnR$/ (seq_extend u)).
\end{minted}
These two notions of sums are related via \coq{Rrseq}, and so are the respective additions of positive (resp. extended positive) reals and the respective pointwise additions of sequences. Once \Trocq{} is informed of these relations, the tactic is able to transfer the statement from the much easier variant on extended reals:
\begin{minted}{coq}
(* relating type $\nnR$ and $\xnnR$ and their respective equalities *)
Trocq Use Param01_paths Param42b_nnR.
(* relating sequence types, sums, addition, addition of sequences*)
Trocq Use Param4a_rseq R_sum_xnnR R_add_xnnR seq_nnR_add.

Lemma sum_xnnR_add : /$\forall$/ (u v : /$\xnnR$/), /$\Sigma_\xnnR$/ (u + v) = /$\Sigma_\xnnR$/ u + /$\Sigma_\xnnR$/ v.
Proof.(...) Qed. (* easy, as no convergence proof is needed *)

Lemma sum_nnR_add : /$\forall$/ (u v : /$\nnR$/), /$\Sigma_\nnR$/ (u + v) = /$\Sigma_\nnR$/ u + /$\Sigma_\nnR$/ v.
Proof. trocq; exact sum_xnnR_add. Qed.
\end{minted}

\vspace{-2ex}
\subsection{Polymorphic, dependent types}\label{ssec:ex-poly-dep}
\paragraph{Polymorphic parameters \trocqcode{theories/Param_list.v}{(code)}.}

Suppose we want to transfer a goal involving lists along an equivalence between the types of the values contained in the lists.
We first prove that the \coq{list} type former is equivalent to itself, and register this fact:
\begin{minted}{coq}
listR : /$\forall$/ A A' (AR : Param44.Rel A A'), Param44.Rel (list A) (list A')
Trocq Use listR.
\end{minted}
We also need to relate with themselves all operations on type \coq{list} involved in the goal, including constructors, and to register these facts, before \Trocq{} is able to transfer any goal, \eg about \coq{list N} to its analogue on \coq{list /$\mathbb{N}$/}.

Note that lemma \coq{listR} requires an \emph{equivalence} between its parameters. If this does not hold, as in the case of type \coq{int} and \coq{Zmodp} from Section~\ref{ssec:ex-iso}, \trocqcode{examples/stuck.v}{the translation is stuck}: weakening does not apply here. In order to avoid stuck translation, we need several versions of \coq{listR} to cover all cases. For instance, the following lemma 
is required for proof transfers from \coq{list Zmodp} to \coq{list int}.
\begin{minted}{coq}
listR2a4 : /$\forall$/ A A' (AR : Param2a4.Rel A A'),
  Param2a4.Rel (list A) (list A').
\end{minted}
\paragraph{Dependent and polymorphic types \trocqcode{examples/Vector_tuple.v}{(code)}.}

Fixed-size vectors can be represented by iterated tuples, an alternative to the
inductive type \coq{Vector.t}, from \Coq{}'s standard library, as follows.
\begin{minted}{coq}
Definition tuple (A : Type) : nat -> Type := fix F n :=
  match n with O => Unit | S n' => F n' * A end.
\end{minted}
On the following mockup example, \Trocq{} transfers a lemma on \coq{Vector.t} to its analogue on \coq{tuple}, about a function \coq{head : /$\forall$/ A n, tuple A (S n) -> A}, and a function \coq{const : /$\forall$/ A, A -> /$\forall$/ n, tuple A n} creating a constant vector, and simultaneously refines integers into the integers modulo $p$ from Section~\ref{ssec:ex-iso}:

\begin{minted}{coq}
Lemma head_cst (n : nat) (i : int): Vector.hd (Vector.const i (S n)) = i.
Proof. destruct n; simpl; reflexivity. Qed. (* easy proof *)

Lemma head_cst' : /$\forall$/ (n : nat) (z : Zmodp), head (const z (S n)) = z.
Proof. trocq. exact head_const. Qed.
\end{minted}
This automated proof only requires proving (and registering) that \coq{head} and \coq{const} are related to their analogue \coq{Vector.hd} and \coq{Vector.const}, from \Coq{}'s standard library. Note that the proof uses the equivalence between \coq{Vector.t} and \coq{tuple} but only requires a retraction between parameter types.

\vspace{-2ex}
\section{Conclusion}
\label{sec:rel-work-and-concl}

The \Trocq{} framework can be seen as a generalization of the univalent parametricity translation~\cite{univparam2}. It allows for weaker relations than equivalence, thanks to a fine-grained control of the data propagated by the translation. 
This analysis is enabled by skolemizing the usual symmetrical presentation of equivalence, so as to expose the data, and by introducing a hierarchy of algebraic structures for relations. This scrutiny allows in particular to get rid of the univalence axiom for a larger class of equivalence proofs~\cite{univparam}, and to deal with refinement relations for arbitrary terms, unlike the \CoqEAL{} library~\cite{DBLP:conf/cpp/CohenDM13}. Altenkirch and Kaposi already proposed a symmetrical, skolemized phrasing of type equivalence~\cite{DBLP:conf/types/AltenkirchK15}, but for different purposes. In particular, they did not study the resulting hierarchy of structures. Definition~\ref{def:umap} however slightly differs from theirs: by reducing the amount of transport involved, it eases formal proofs significantly in practice, both in the internal library of \Trocq{} and for end-users of the tactic.

The concrete output of this work is a plugin~\cite{trocq_zenodo} that consists of about 3000~l. of original \Coq{} proofs and 1200 l. of meta-programming, in the \Elpi{} meta-language, excluding white lines and comments. This plugin goes beyond the state of the art in two ways. First, it demonstrates that a single implementation of this parametricity framework covers the core features of several existing other tactics, for refinements~\cite{DBLP:conf/cpp/CohenDM13,10.1007/978-3-642-32347-8_7}, generalized rewriting~\cite{DBLP:journals/jfrea/Sozeau09}, and proof transfer~\cite{univparam2}. Second, it addresses use cases, such as Example~\ref{ex:sums}, that are beyond the skills of any existing tool in any proof assistant based on type theory.  The prototype plugin arguably needs an improved user interface so as to reach the maturity of some of the aforementioned existing tactics. It would also benefit from an automated generation of equivalence proofs, such as \PumpkinPi{}~\cite{DBLP:conf/pldi/RingerPYLG21}.


\begin{credits}
  \ackname{
    The authors would like to thank András Kovács, Kenji Maillard, Enrico Tassi, Quentin Vermande, Théo Winterhalter, and anonymous reviewers, whose comments greatly helped us improve this article.
    }
\end{credits}

\bibliographystyle{splncs04}
\bibliography{main}

\end{document}